\documentclass[11pt]{article}
\usepackage{fullpage}
\usepackage{pbox}
\usepackage{url,amsfonts, amsmath, amssymb, amsthm}
\usepackage[pdftex]{graphicx}
\usepackage{color}
\usepackage{enumitem}
\theoremstyle{plain}
\newtheorem{theorem}{Theorem}[section]
\newtheorem{lemma}[theorem]{Lemma}

\newtheorem{conjecture}[theorem]{Conjecture}

\newtheorem{observation}[theorem]{Observation}
\theoremstyle{definition}
\newtheorem{definition}[theorem]{Definition}

\newtheorem{remark}[theorem]{Remark}

\newcommand{\dist}{\operatorname{dist}}

\newcommand{\E}{{\mathbb E\/}}
\newcommand{\paren}[1]{\mathopen{}\left( #1 \right)\mathclose{}}

\newcommand{\ceil}[1]{\lceil #1 \rceil}
\newcommand{\floor}[1]{\lfloor #1 \rfloor}
\newcommand{\f}[2]{\frac{#1}{#2}}
\newcommand{\fr}[2]{\mbox{$\frac{#1}{#2}$}}

\newcommand{\bydef}{\stackrel{\operatorname{def}}{=}}
\newcommand{\poly}{\operatorname{poly}}
\newcommand{\eps}{\epsilon}

\newcommand{\ignore}[1]{}

\newcommand{\rb}[2]{\raisebox{#1 mm}[0mm][0mm]{#2}}
\newcommand{\istrut}[2][0]{\rule[- #1 mm]{0mm}{#1 mm}\rule{0mm}{#2 mm}}
\newcommand{\zero}[1]{\makebox[0mm][l]{$#1$}}

\newcommand{\hcm}[1][1]{\hspace*{#1 cm}}

\newcommand{\Base}{\dot{B}}
\newcommand{\DblBase}{\ddot{B}}
\newcommand{\nBase}{\dot{n}}
\newcommand{\nDblBase}{\ddot{n}}
\newcommand{\mBase}{\dot{m}}
\newcommand{\mDblBase}{\ddot{m}}
\newcommand{\Pairs}{\mathcal{P}}
\newcommand{\lab}{\operatorname{label}}
\newcommand{\Labels}{\mathcal{L}}
\newcommand{\Graph}{\mathcal{H}}
\newcommand{\RevGraph}{\overline{\Graph}}
\newcommand{\Layer}[1]{\ddot{L}_{#1}}
\newcommand{\Ball}{\mathcal{B}}
\newcommand{\Succ}[2]{C_{#1}(#2)}
\newcommand{\Fail}[2]{I_{#1}(#2)}

\newcommand{\Span}{{\sc Span.}}
\newcommand{\Emul}{{\sc Emul.}}
\newcommand{\all}{{\sc All}}
\newcommand{\approxdist}{\widetilde{\dist}}
\newcommand{\Loss}{\xi}

\newcommand{\Bollobas}{Bollob\'{a}s}

\newcommand{\Althofer}{Alth\"{o}fer}
\newcommand{\Matousek}{Matou\v{s}ek}

\usepackage[dvipsnames]{xcolor}

\begin{document}

\title{A Hierarchy of Lower Bounds for Sublinear Additive Spanners\thanks{Supported
by NSF grants CCF-1217338, CNS-1318294, CCF-1417238, CCF-1514339, CCF-1514383, CCF-1637546,
and BSF Grant 2012338.  Email: \texttt{abboud@cs.stanford.edu}, \texttt{gbodwin@cs.stanford.edu}, \texttt{pettie@umich.edu}.  A preliminary version of this paper will appear in the conference proceedings of SODA 2017.}}

\author{
Amir Abboud\\
Stanford University
\and
Greg Bodwin\\
Stanford University
\and
Seth Pettie\\
University of Michigan}

\date{}

\maketitle
\thispagestyle{empty}
\setcounter{page}{0}

\begin{abstract}
Spanners, emulators, and approximate distance oracles
can be viewed as {\em lossy} compression schemes that represent an 
unweighted graph metric in small 
space, say $\tilde{O}(n^{1+\delta})$ bits.  There is an inherent tradeoff
between the sparsity parameter $\delta$ and the {\em stretch function} $f$ of the compression scheme,
but the qualitative nature of this tradeoff has remained a persistent open problem.

It has been known for some time that when $\delta\ge 1/3$ there are 
schemes with constant {\em additive} stretch (distance $d$ is stretched to at most $f(d) = d + O(1)$), 
and recent results of Abboud and Bodwin show that when $\delta < 1/3$ there are no such schemes.
Thus, to get practically efficient graph compression with $\delta \to 0$ we must pay super-constant additive stretch, but exactly how much do we have to pay?

In this paper we show that the lower bound of Abboud and Bodwin is just the first step in a \emph{hierarchy} of lower bounds
that characterize the asymptotic behavior of the optimal stretch function $f$ for sparsity parameter $\delta \in (0,1/3)$.
Specifically, for any integer $k\ge 2$, any compression scheme with size $O(n^{1+\f{1}{2^k-1} - \epsilon})$
has a {\em sublinear additive stretch} function $f$:
\[
f(d) = d + \Omega(d^{1-\f{1}{k}}).
\]  
This lower bound matches Thorup and Zwick's (2006) construction of sublinear additive {\em emulators}.
It also shows that Elkin and Peleg's $(1+\epsilon,\beta)$-spanners have an essentially optimal tradeoff between $\delta,\epsilon,$ and $\beta$,
and that the sublinear additive spanners of Pettie (2009) and Chechik (2013) are not too far from optimal.
To complement these lower bounds we present a new construction of
$(1+\epsilon, O(k/\epsilon)^{k-1})$-spanners
with size $O((k/\epsilon)^{h_k} kn^{1+\f{1}{2^{k+1}-1}})$, where $h_k < 3/4$.
This size bound improves on the spanners of Elkin and Peleg (2004), Thorup and Zwick (2006), and Pettie (2009).
According to our lower bounds neither the size nor stretch function can be substantially improved.

Our lower bound technique exhibits several interesting degrees of freedom in the framework of Abboud and Bodwin.
By carefully exploiting these freedoms, we are able to obtain lower bounds for several related combinatorial objects.
We get lower bounds on the size of $(\beta,\epsilon$-\emph{hopsets}, matching Elkin and Neiman's construction (2016),
and lower bounds on {\em shortcutting sets} for digraphs that preserve the transitive closure.  Our lower bound simplifies
Hesse's (2003) refutation of Thorup's conjecture (1992), which stated that adding a linear number of shortcuts suffices to reduce the diameter to polylogarithmic.
Finally, we show matching upper and lower bounds for graph compression schemes that work for graph metrics with girth at least $2\gamma+1$.
One consequence is that 
Baswana et al.'s (2010) additive $O(\gamma)$-spanners with size $O(n^{1+\f{1}{2\gamma+1}})$ cannot be improved in the exponent.
\end{abstract}

\newpage

\section{Introduction}

{\em Spanners}~\cite{PS89}, {\em emulators}~\cite{DHZ00,TZ06}, and {\em approximate distance oracles}~\cite{TZ05}
can be viewed as kinds of compression schemes that approximately encode the distance metric of a (dense) 
undirected input graph $G=(V,E)$ in small space, where the notion of {\em approximation} 
is captured by a non-decreasing {\em stretch function} $f\,:\, \mathbb{N}\rightarrow\mathbb{N}$.

\begin{description}
\item[Spanners.] An $f(d)$-{\em spanner} $G' = (V,E')$ is a subgraph of $G$ for which $\dist_{G'}(u,v)$ is at most $f(\dist_G(u,v))$.
An $(\alpha,\beta)$-spanner is one with stretch function $f(d) = \alpha d + \beta$.  Notable special cases include {\em multiplicative $\alpha$-spanners}~\cite{PS89,Althofer+93,EP04,TZ05,BaswanaS07,BKMP10}, 
when $\beta=0$,
and {\em additive $\beta$-spanners}~\cite{ACIM99,DHZ00,EP04,TZ06,BKMP10,Woodruff10,Chechik13,Knudsen14}, 
when $\alpha=1$.  See~\cite{EP04,BKMP10,TZ06,Pettie-Span09,Chechik13,Parter14} for ``mixed'' 
spanners with $\alpha > 1, \beta>0$.

\item[Emulators.] An $f(d)$-{\em emulator} (also called a {\em Steiner spanner}~\cite{Althofer+93}) is a {\em weighted} graph $G' = (V'\supseteq V, E',w')$ such that for each $u,v\in V$, 
$\dist_{G'}(u,v) \in [\dist_G(u,v), f(\dist_G(u,v))]$.  
In other words, one is allowed to add Steiner points ($V'\backslash V$) and long-range (weighted) edges $(u,v) \in E'\backslash E$
such that distances are non-contracting.

\item[(Unconstrained) Distance Oracles.] For our purposes, an $f(d)$-approximate distance oracle using space $s$ is a bit string 
in $\{0,1\}^{s}$
such that given $u,v\in V$, 
an estimate $\approxdist(u,v) \in [\dist_G(u,v), f(\dist_G(u,v))]$ can be computed by examining only the bit string.
Note: the term ``oracle'' was used in~\cite{TZ05} to indicate that 
$\approxdist(u,v)$ is computed in constant time~\cite{PatrascuR14,AbrahamG11,Chechik15}. 
Later work considered distance oracles with non-constant query time~\cite{PoratR13,AgarwalG13,Agarwal14,ElkinP16}.
In this paper we make no restrictions on the query time at all.  Thus, for our purposes distance oracles 
generalize spanners, emulators, and related objects.
\end{description}

In this paper we establish essentially optimal tradeoffs between the size of the compressed graph representation
and the asymptotic behavior of its stretch function $f$.
In order to put our results in context we must recount the developments of the last 30 years that
investigated multiplicative, additive, $(\alpha,\beta)$, and sublinear additive stretch functions.

\subsection{Multiplicative Stretch}
Historically, the first notion of stretch studied in the literature was purely multiplicative stretch.  \Althofer{} et al.~\cite{Althofer+93} 
quickly settled the problem
by showing that any graph contains an $\alpha$-spanner with at most $m_{\alpha+2}(n)$ edges, and that the claim is false for $m_{\alpha+2}(n)-1$. 
Here $m_g(n)$ is the maximum number of edges in a graph with $n$ vertices and girth $g$.
The upper bound of~\cite{Althofer+93} follows directly from the observation that a natural greedy construction 
never closes a cycle with length at most $\alpha+1$; the lower bound follows from the fact that no strict subgraph of a graph with 
girth $\alpha+2$ is an $\alpha$-spanner.\footnote{Removing any edge stretches the distance between its endpoints from 1 to at least $\alpha+1$.
Moreover, since every graph contains a bipartite subgraph with at least half the edges, $m_{2k+1} \le 2m_{2k+2}(n)$ for every $k$.
Thus, there are $(2k-1)$-spanners with size $O(m_{2k+2}(n))$.}
It has been conjectured~\cite{Erdos63,BondyS74,Bollobas78}
that the trivial upper bound $m_{2k+1}(n), m_{2k+2}(n) = O(n^{1+1/k})$ is sharp up to the leading constant, 
but this {\em Girth Conjecture} has only been proved for $k=1$ (trivial), and $k\in\{2,3,5\}$~\cite{Brown66,ErdosRS66,Reiman58,Wenger91,Tits59,Benson66,LazebnikU93}.
See~\cite{LazebnikUW95,LazebnikUW96,WoldarU93} for lower bounds on $m_g(n)$.

\subsection{Additive Stretch}
The Girth Conjecture implies that a spanner with size $O(n^{1+1/k})$ must stretch some pair of adjacent vertices at original distance
$d=1$ to distance $2k-1$.   If ``stretch'' is defined \emph{a priori} to be multiplicative, then such $(2k-1)$-spanners are optimal.  However, there is no reason to believe that $f(d) = (2k-1)d$ is an optimal stretch \emph{function} for size $O(n^{1+1/k})$.  The girth argument could also be interpreted as lower bounding additive stretch or $(\alpha,\beta)$-stretch.  In general, the Girth Conjecture only implies that $(\alpha,\beta)$-spanners with size $O(n^{1+1/k})$
have $\alpha+\beta \ge 2k-1$.

Aingworth, Chekuri, Indyk, and Motwani~\cite{ACIM99} gave a construction of an additive $2$-spanner with size $\tilde{O}(n^{3/2})$, which is optimal in the sense that 
neither the additive stretch $2$ nor exponent $3/2$ can be unilaterally improved.\footnote{Moreover, 
later results of \Bollobas{} et al.~\cite{BCE06} show that for spanner size $O(n^{3/2})$,
the stretch {\em function} $f(d) = d+2$ is optimal for $1\le d \le \Theta(\sqrt{n})$.
See~\cite{EP04,TZ06,BKMP10,Knudsen14} for constructions of additive 2-spanners with size $O(n^{3/2})$.}
This result raised the tantalizing possibility that there exist arbitrarily sparse additive spanners.  
Dor, Halperin, and Zwick~\cite{DHZ00} observed that additive $4$-{\em emulators} exist with size $\tilde{O}(n^{4/3})$, 
i.e., the emulator introduces weighted edges connecting distant vertex pairs.
Baswana, Kavitha, Mehlhorn, and Pettie~\cite{BKMP10} constructed additive 6-spanners
with size $O(n^{4/3})$ and Chechik~\cite{Chechik13} constructed additive-4 spanners with size $\tilde{O}(n^{7/5})$.
See~\cite{Woodruff10,Knudsen14,EP04,TZ06,DHZ00,BKMP10} for other constructions of additive 2- and 6-spanners.

The ``$4/3$'' exponent proved to be very resilient, for both emulators and spanners with additive stretch. 
This led to a line of work establishing additive spanners below the $n^{4/3}$ threshold with stretch polynomial in $n$~\cite{BCE06,BKMP10,Pettie-Span09,Chechik13,BodwinW15}.
The additive spanners of Bodwin and Williams~\cite{BodwinW16} with 
stretch function $f(d) = d + n^\epsilon$ have size that is the minimum of 
$O(n^{\f{4}{3} - \f{7\epsilon}{9} + o(1)})$ and $O(n^{\f{5}{4} - \f{5\epsilon}{12}+o(1)})$.

\subsection{Sublinear Additive Stretch}

Elkin and Peleg~\cite{EP04} showed that the ``4/3 barrier'' could also be broken by tolerating $1+\epsilon$ multiplicative stretch. 
In particular, for any integer $\kappa$ and real $\epsilon>0$, there are $(1+\epsilon,\beta)$-spanners with size
$O(\beta n^{1+1/\kappa})$, where $\beta = O(\epsilon^{-1}\log \kappa)^{\log \kappa}$.  The construction algorithm and size-bound
both depend on $\epsilon$.  Thorup and Zwick~\cite{TZ06} gave a surprisingly simple construction of
an $O(kn^{1+\f{1}{2^{k+1}-1}})$-size {\em emulator} with $(1+\epsilon, O(k/\epsilon)^{k-1})$-type stretch.

Thorup and Zwick's emulator has the special property that its stretch holds for every $\epsilon>0$ \emph{simultaneously}, i.e., it can be selected as a function of $d$.  
Judiciously choosing $\epsilon = k/d^{\f{1}{k}}$ leads 
to an emulator with a {\em sublinear additive stretch} 
function $f(d) = d + O(kd^{1-\frac{1}{k}} + 3^k)$.\footnote{The Thorup-Zwick 
emulator can easily be converted to a $(1+\epsilon,\beta)$-spanner by replacing weighted
edges with paths up to length $\beta$.  A careful analysis shows the size of the resulting spanner can be made 
$O((k/\epsilon)^{O(1)} n^{1+\f{1}{2^{k+1}-1}})$ (see Section~\ref{sect:new-upper-bounds})
which would slightly improve on~\cite{EP04}.  
Elkin [personal communication, 2013] has stated that with minor changes, the Elkin-Peleg~\cite{EP04} spanners
can also be expressed as $(1+\epsilon, O(k/\epsilon)^{k-1})$-spanners with size $O((k/\epsilon)^{O(1)}n^{1+\f{1}{2^{k+1}-1}})$.
We state these bounds in Figure~\ref{fig:prior-work} rather than those of~\cite{EP04} 
in order to facilitate easier comparisons with subsequent constructions
\cite{TZ06,Chechik13,Pettie-Span09}, and the new constructions of Section~\ref{sect:new-upper-bounds}.}
Thorup and Zwick also showed that this same stretch function also applies to 
their earlier~\cite{TZ05} construction of multiplicative $(2k+1)$-spanners with size $O(kn^{1+\f{1}{k+1}})$.
Pettie~\cite{Pettie-Span09} gave a construction of sublinear additive {\em spanners} whose size-stretch tradeoff is closer
to the Thorup-Zwick emulators.  For stretch function $d+O(kd^{1-\f{1}{k}} + 3^k)$ the size is 
$O(kn^{1+\f{(3/4)^{k-2}}{7 - 2(3/4)^{k-2}}})$, which is always $o(n^{1+ (3/4)^{k+3}})$ for any fixed $k$.
At their sparsest, Thorup and Zwick's emulators~\cite{TZ06} and Pettie's spanners~\cite{Pettie-Span09} have size 
$O(n\log\log n)$ and stretch 
$f(d) = d + O(\log\log n) \cdot d^{1-\Theta(1/\log\log n)} + (\log n)^{\log_2 3}$.  
Pettie~\cite{Pettie-Span09} gave an even sparser 
$(1+\epsilon,O(\epsilon^{-1}\log\log n)^{\log\log n})$-spanner 
with size $O(n\log\log(\epsilon^{-1}\log\log n))$.

\subsection{Lower Bounds}

Woodruff proved that any $k^{-1}n^{1+1/k}$-size spanner with stretch function $f$ must have $f(k)\ge 3k$.
As a corollary, additive $(2k-2)$-spanners must have size $\Omega(k^{-1}n^{1+1/k})$, independent of the status of the Girth Conjecture.
\Bollobas, Coppersmith, and Elkin~\cite{BCE06} showed that if the stretch $f$ is such that $f(d) = d$ for $d\ge D$, 
then $\Omega(n^2/D)$-size is necessary and sufficient for spanners and emulators.

In a recent surprise, Abboud and Bodwin~\cite{AbboudB16} proved that no additive $\beta$-spanners, emulators, nor distance oracles exist with $\beta=O(1)$ and exponent less than $4/3$.
More precisely, any construction of these three objects with additive $\beta=O(1)$ stretch has size 
$\Omega(n^{4/3}/2^{O(\sqrt{\log n})})$ and any construction 
with size $O(n^{4/3-\epsilon})$ has additive stretch $\beta = n^{\delta}$ for some $\delta=\delta(\epsilon)$.
This result explained why all prior additive spanner constructions had a strange transition at $4/3$~\cite{DHZ00,TZ06,BKMP10,Chechik13,BodwinW16,Knudsen14,Woodruff10},
but it did not suggest what the {\em optimal} stretch function should be for sparsity $n^{1+\delta}$ when $\delta\in[0,1/3)$.

\begin{figure}
\centering
\scalebox{.95}{
\begin{tabular}{|l|c|c|c||c|}
\multicolumn{1}{l}{}	& \multicolumn{4}{c}{\rb{1}{\bf\large Stretch Function}}\\\cline{2-5}
\multicolumn{1}{c|}{} 		& $d + O\paren{\sqrt{d}}$			& $d + O\paren{d^{\f{2}{3}}}$				&  $d + O\paren{d^{\f{3}{4}}}$ 				& $d + O\paren{kd^{1-\f{1}{k})}}$\\
\multicolumn{1}{c|}{}			&	or	& 	or 	&     or	&    or\\
\multicolumn{1}{c|}{\rb{1}{\bf\large Citation}}  & $\paren{1+\epsilon,O\paren{\fr{1}{\epsilon}}}$  & $\paren{1+\epsilon,O\paren{\fr{1}{\epsilon}}^2}$	& $\paren{1+\epsilon,O\paren{\fr{1}{\epsilon}}^3}$	& $\paren{1+\epsilon,O\paren{\fr{k}{\epsilon}}^{k-1}}$\istrut[2]{0}\\\cline{1-5}
Elkin \& Peleg   \hfill \Span 	& $O\paren{\epsilon^{-O(1)}n^{\f{8}{7}}}$	&	$O\paren{\epsilon^{-O(1)}n^{\f{16}{15}}}$	&	$O\paren{\epsilon^{-O(1)}n^{\f{32}{31}}}$	&	$O\paren{\paren{\fr{k}{\epsilon}}^{O(1)} n^{1+\f{1}{2^{k+1}-1}}}$\istrut[4]{6}\\\hline
\rb{-1}{Thorup} \hfill \Emul	& $O\paren{n^{\f{8}{7}}}$		&	$O\paren{n^{\f{16}{15}}}$				& $O\paren{n^{\f{32}{31}}}$				& $O\paren{kn^{1+\f{1}{2^{k+1}-1}}}$\istrut[3]{6}\\
\rb{1}{\& Zwick}	\hfill \Span		& $O\paren{n^{\f{4}{3}}}$		&	$O\paren{n^{\f{5}{4}}}$				& $O\paren{n^{\f{6}{5}}}$					& $O\paren{kn^{1+\f{1}{k+1}}}$\istrut[4]{0}\\\hline
Pettie 	 \hfill \Span 	& $O\paren{n^{\f{6}{5}}}$		& 	$O\paren{n^{\f{25}{22}}}$				& $O\paren{n^{\f{103}{94}}}$				& $O\paren{kn^{1+\frac{(3/4)^{k-2}}{7-2(3/4)^{k-2}}}}$\istrut[6]{8}\\\hline
Chechik	\hfill \Span 	& $\tilde{O}\paren{n^{\f{20}{17}}}$\istrut[3]{5}	& \multicolumn{3}{c|}{} \\\hline

{\bf New} 			\hfill \Span & $O\paren{\epsilon^{-\f{2}{7}}n^{\f{8}{7}}}$ &	$O\paren{\epsilon^{-\f{7}{15}}n^{\f{16}{15}}}$			& $O\paren{\epsilon^{-\f{18}{31}}n^{\f{32}{31}}}$		
																										& $O\paren{\paren{\fr{k}{\epsilon}}^h kn^{1+\f{1}{2^{k+1}-1}}}$\istrut[4]{6}\\\hline

\pbox{20cm}{\bf New Lower\\Bounds} \hfill \all	& $\Omega\paren{n^{\f{4}{3}-o(1)}}$ & 	$\Omega\paren{n^{\f{8}{7}-o(1)}}$		& $\Omega\paren{n^{\f{16}{15}-o(1)}}$	& $\Omega\paren{n^{1+\f{1}{2^{k}-1}-o(1)}}$\istrut[4]{6}\\\hline

\end{tabular}
}
\caption{\label{fig:prior-work}
A summary of spanners and emulators with 
$(1+\epsilon, O(k/\epsilon)^{k-1})$-type stretch and
sublinear additive stretch $d+O(kd^{1-\f{1}{k}})$.
Note: the new lower bounds do not contradict the upper bounds;
the lower bounds are for stretch functions with smaller leading constants in the 
$O(k/\epsilon)^{k-1}$ and $O(kd^{1-\f{1}{k}})$ terms.
In the last cell of the table, $h = \frac{3\cdot 2^{k-1} - (k+2)}{2^{k+1}-1} < 3/4$, which improves
the dependence on $\epsilon$ that can be obtained from modified versions of existing constructions~\cite{EP04,TZ06}.}
\end{figure}

\subsection{New Results}

\paragraph{Distance Oracle Lower Bounds.}

Our main result is a hierarchy of lower bounds for spanners, emulators, and distance oracles, which shows that 
tradeoffs offered by Thorup and Zwick's~\cite{TZ06} sublinear additive emulators~\cite{TZ06} 
and Elkin and Peleg's $(1+\epsilon,\beta)$-spanners cannot be substantially improved.  
Building on Abboud and Bodwin's~\cite{AbboudB16} $\Omega(n^{4/3-o(1)})$ lower bounds for additive spanners, 
we prove that for every integer $k\ge 2$ and $d < n^{o(1)}$,
there is a graph $\Graph_k$ on $n$ vertices
and $n^{1+\f{1}{2^{k}-1}-o(1)}$ edges such that any spanner with size $n^{1+\f{1}{2^{k}-1} - \epsilon}$, $\epsilon>0$,
stretches vertices at distance $d$ to at least $d + c_k d^{1-\f{1}{k}}$ for a constant $c_k =\Theta(1/k)$.
More generally, we exhibit graph families that cannot be \emph{compressed} into distance oracles on $n^{1+\f{1}{2^{k}-1} - \epsilon}$ bits such that distances can be recovered below this error threshold.
The consequences of this construction are that
the existing sublinear additive emulators~\cite{TZ06}, sublinear additive spanners~\cite{Pettie-Span09,Chechik13}, 
and $(1+\epsilon,\beta)$-spanners~\cite{EP04,TZ06,Pettie-Span09} are, to varying degrees, close to optimal.  Specifically,
\begin{itemize}
\item 
The $(d + O(kd^{1-\f{1}{k}} + 3^k))$-emulator~\cite{TZ06} with size $O(n^{1+\f{1}{2^{k+1}-1}})$ cannot be improved by more than a constant
factor in the stretch $O(kd^{1-\f{1}{k}})$, or by a $o(1)$ in the exponent $1+\f{1}{2^{k+1}-1}$.  

\item The sublinear additive {\em spanners} of Pettie~\cite{Pettie-Span09} and Chechik~\cite{Chechik13} probably have suboptimal exponents, but not by much.
For example, the exponent of Chechik's~\cite{Chechik13} $\tilde{O}(n^{20/17})$-size $(d+O(\sqrt{d}))$-spanner is within 0.034 of optimal and
the exponent of Pettie's~\cite{Pettie-Span09} $O(n^{25/22})$-size $(d+O(d^{2/3}))$-spanner is within 0.07 of optimal.

\item When $\epsilon \ge 1/n^{o(1)}$, the existing constructions of $(1+\epsilon, O(k/\epsilon)^{k-1})$-spanners~\cite{EP04,TZ06,Pettie-Span09} 
with size $O\paren{(k/\epsilon)^{O(1)}n^{1+\f{1}{2^{k+1}-1}}}$ cannot be substantially improved in either the additive $O(k/\epsilon)^{k-1}$ term or the exponent $1+\f{1}{2^{k+1}-1}$.
This follows from the fact that any spanner with stretch of type $(1+\hat{\epsilon},O(k/\hat{\epsilon})^{k-1})$, for every $\hat{\epsilon} \ge \epsilon$ functions as a $(d+O(kd^{1-\f{1}{k}}))$-spanner
for distances $d \le O(k/\epsilon)^{k}$.
However, there is no reason to believe that the {\em size} 
of such $(1+\epsilon,\beta)$-spanners must depend on $\epsilon$, as it does in the current constructions.
\end{itemize}

There is an interesting new hierarchy of {\em phase transitions} in the interplay between our lower bounds previous upper bounds~\cite{TZ06}.
Let $C$ be a sufficiently large constant and $c$ be a sufficiently small constant.
If one wants a graph compression scheme with stretch $f(d) = d + C\sqrt{d}$, 
then one needs only $\widetilde{O}(n^{8/7})$ bits of space to store an emulator~\cite{TZ06}.
However, if we want a slightly improved stretch $f(d) = d + c\sqrt{d}$, then, by our lower bound,
the space requirement leaps to $\Omega(n^{4/3 - o(1)})$.
In general, the optimal space for stretch function $f(d) = d + c' d^{1 - 1/k}$ takes a polynomial jump as we shift 
$c'$ from some sufficiently large constant $O(k)$ to a sufficiently small constant $\Omega(1/k)$.

An important take-away message from our work is that the sublinear additive stretch functions of type $f(d) = d + O(d^{1 - 1/k})$ used by Thorup and Zwick~\cite{TZ06} 
are exactly of the ``right'' form.  For example, such plausible-looking stretch functions as $f(d) = d + O(d^{1/3})$ 
and $f(d) = d + O(d^{2/3} / \log d)$ could only exist in the narrow bands not covered by our lower bounds:
between space $n^{4/3-o(1)}$ and $n^{4/3}$ and between space $n^{8/7-o(1)}$ and $n^{8/7}$.

\paragraph{Spanner Upper Bounds.}

To complement our lower bounds we provide new upper bounds on the sparsity of spanners with stretch of type
$(1+\hat{\epsilon},O(k/\hat{\epsilon})^{k-1})$, which holds for every $\hat{\epsilon} \ge \epsilon$.
Our new spanners have size $O((k/\epsilon)^{h}kn^{1+\f{1}{2^{k+1}-1}})$, where $h=\f{3\cdot 2^{k-1} - (k+2)}{2^{k+1}-1} < 3/4$.
This construction improves on the bounds that can be derived from~\cite{TZ06,EP04,Pettie-Span09} in the dependence on 
$\epsilon$.\footnote{No bounds of this type are stated explicitly in~\cite{TZ06} or~\cite{EP04}.
In order to get a bound of this type---with the $1+\f{1}{2^{k+1}-1}$ exponent
and some $\poly(1/\epsilon)$ dependence on $\epsilon$---
one must only adjust the sampling probabilities of~\cite{TZ06};
however, adapting \cite{EP04} requires slightly more significant changes [Elkin, personal communication, 2013].}
For example, one consequence of this result is an $O(D^{1/7}n^{8/7})$-size spanner that functions 
as a $(d+O(\sqrt{d}))$-spanner for all $d\le D$.  This size bound is an improvement on 
Chechik's $(d+O(\sqrt{d}))$-spanner, as long as $D < n^{4/17}$.

\paragraph{Hopset Lower Bounds.}
Hopsets are fundamental objects that are morally similar to emulators.
They were explicitly defined by Cohen \cite{Cohen00} but used implicitly in many earlier works \cite{UY91,KS97,Cohen97,SS99}.
Let $G = (V, E, w)$ be an arbitrary undirected \emph{weighted} graph and $H \subset {V\choose 2}$ be a set of edges called the \emph{hopset}.  
In the united graph $G' = (V, E\cup H,w)$, the weight of an edge $(u,v)\in H$ is the length of the shortest path in $G$ between $u$ and $v$.
Define the \emph{$\beta$}-limited distance in $G'$, denoted $\dist^{(\beta)}_{G'}(u, v)$, 
to be the length of the shortest path from $u$ to $v$ that uses at most $\beta$ edges in $G'$.\footnote{Note that whereas $\dist = \dist^{(\infty)}$ is metric, $\dist^{(\beta)}$ does not necessarily satisfy the triangle inequality for finite $\beta$.}
We call $H$ a \emph{$(\beta,\eps)$-hopset}, where $\beta \ge 1, \epsilon>0$, 
if, for any $u, v \in V$, we have
\[
\dist_{G'}^{(\beta)}(u, v) \le (1 + \eps) \dist_G(u, v).
\]
There is clearly some three-way tradeoff between $\beta,\eps,$ and $|H|$.
Elkin and Neiman \cite{EN16} recently showed that any graph has a $(\beta,\eps)$-hopset with size $\tilde{O}(n^{1+1/\kappa})$, where 
$\beta = O\left( \frac{\log \kappa}{\eps} \right)^{\log \kappa}$.\footnote{It is likely that Elkin and Neiman's tradeoff could be more precisely stated as follows:
for any positive integer $k$ and $\epsilon > 0$, 
there is an $\tilde{O}(n^{1+\frac{1}{2^{k+1}-1}})$ size $(\beta,\eps)$-hopset with $\beta = O(k/\epsilon)^k$.}

In this work, we show that any construction of $(\beta, \eps)$-hopsets with worst-case 
size $n^{1 + \frac{1}{2^{k} - 1} - \delta}$, where $k\ge 1$ is an integer and $\delta > 0$, 
must have $\beta = \Omega_k\left(\frac{1}{\eps}\right)^{k}$.
For example, hopsets with $\beta = o(1/\epsilon)$ must have size $\Omega(n^{2-o(1)})$ 
and those with $\beta = o(1/\epsilon^2)$ must have size $\Omega(n^{4/3-o(1)})$.
This essentially matches the Elkin-Neiman tradeoff, up to a constant in $\beta$ that depends on $k$.

\paragraph{Lower Bounds on Shortcutting Digraphs.}
In 1992, Thorup~\cite{Thorup92} conjectured that the diameter of any directed graph $G = (V, E)$ could be drastically reduced with a small number of \emph{shortcuts}.
In particular, there exists another directed graph $G' = (V,E')$ with $|E'| = O(|E|)$ and the same transitive closure relation as $G$ ($\leadsto$),
such that if $u\leadsto v$, then there is a $\poly(\log n)$-length path from $u$ to $v$ in $G'$.
Thorup's conjecture was confirmed for trees~\cite{Thorup92,Thorup97-par-shortcut,Chaz87}
and planar graphs~\cite{Thorup95}, but finally refuted by Hesse~\cite{Hesse03} for general graphs.
In this paper we give a simpler 1-page proof of Hesse's refutation by modifying our spanner lower bound construction.

\paragraph{Spanners for High-Girth Graphs.}
Our lower bounds apply to the class of {\em all} undirected graph metrics.
Baswana, Kavitha, Mehlhorn, and Pettie~\cite{BKMP10} gave sparser spanners for a \emph{restricted} class of graph metrics.
Specifically, graphs with girth at least $2\gamma+1$ contain additive $6\gamma$-spanners with size $O(n^{1+\f{1}{2\gamma+1}})$.
We adapt our lower bound construction to prove that the exponent $1+\f{1}{2\gamma+1}$ is optimal, assuming the Girth Conjecture, and more generally we
give lower bounds on compression schemes for the class of graphs with girth at least $2\gamma+1$.
Any scheme that uses $n^{1+\f{1}{(\gamma+1)2^{k-1}-1}-\epsilon}$ bits must have stretch 
$f(d) \ge d + \Omega(d^{1-1/k})$, for any $d < n^{o(1)}$.
We also give new constructions of emulators and spanners for girth-$(2\gamma+1)$ graphs 
that shows that the exponent $1+\f{1}{(\gamma+1)2^{k-1}-1}$ is the best possible.

\subsection{Related Work}

Much of the recent work on spanners has focused on preserving or approximating
distances between specified {\em pairs} of vertices.
See~\cite{CE06,AbboudB16-SODA,AbboudB16} for lower bounds on pairwise spanners
and \cite{CE06,Pettie-Span09,CyganGK13,KavithaV15,Kavitha15,AbboudB16-SODA,Parter14,RTZ05} for upper bounds.
Pairwise spanners have proven to be useful tools for constructing (sublinear) additive spanners;
see~\cite{Pettie-Span09,Chechik13,BodwinW15}.

The space/stretch tradeoffs offered by the best distance 
oracles~\cite{Chechik15,PatrascuR14,PatrascuRT12,AbrahamG11,AgarwalG13,Agarwal14,ElkinP16}
are strictly worse than those of the best spanners and emulators, even though
distance oracles are entirely {\em unconstrained} in how they encode the graph metric.
This is primarily due to the requirement that distance oracles respond to queries quickly.
There are both unconditional~\cite{SommerVY09} 
and conditional~\cite{CohenP10,PatrascuR14,PatrascuRT12} lower bounds suggesting
that distance oracles with reasonable query time cannot match the best spanners or emulators.

\subsection{Organization}
In Section~\ref{sect:lower-bound} we generalize Abboud and Bodwin's construction~\cite{AbboudB16}
to give a spectrum of lower bounds against graph compression schemes 
with sublinear additive stretch and $(1+\epsilon,\beta)$-stretch.
In Section~\ref{sect:new-upper-bounds} we combine ideas from Thorup and Zwick's emulators~\cite{TZ06} and
Pettie's spanners~\cite{Pettie-Span09} to attain a new bound on sparse $(1+\epsilon,\beta)$-spanners.
In Section~\ref{sect:hopset} we prove tight bounds on $(\beta,\epsilon)$-hopsets.
In Section~\ref{sect:girth-lbs} we generalize the construction of Section~\ref{sect:lower-bound} 
to give stretch-sparseness lower bounds on the class of graphs with girth at least $2\gamma+1$.
Matching upper bounds for graphs of gith $2\gamma+1$ are given in Section~\ref{sect:girth-upper-bounds}.
In Section~\ref{sect:shortcut} we give a simpler refutation of Thorup's shortcutting conjecture.
In Section~\ref{sect:conclusion} we highlight some remaining open problems.

\section{The Lower Bound Construction}\label{sect:lower-bound}

The graphs in this section are parameterized by an integer $\ell \ge 2$, 
which determines the length of the hardest shortest paths to approximate.
Each graph has a layered structure, consisting of a layer of {\em input ports},
some number of interior layers, and a layer of {\em output ports}.
In any given graph construction, $p$ is the number of input/output ports.
The construction of $\Base[p]$, $\DblBase[p]$, and $\Graph_2[p]$
is essentially the same as the graphs constructed by 
Abboud and Bodwin~\cite{AbboudB16}.

\subsection{The First Base Graph}

Let $\Base[p] = (L_0 \cup \ldots \cup L_{\ell}, E)$ be an $(\ell+1)$-layer graph with the following properties:
\begin{itemize}
\item $\Base[p]$ has $p$ vertices per layer, and all edges
connect vertices in adjacent layers.
\item Each edge $e$ is assigned a $\lab(e) \in \Labels[p]$. 
For any vertex $u$, the edges connecting $u$ to the previous layer have distinct labels
and the edges connecting $u$ to the subsequent layer have distinct labels.
\item Let $\Pairs(\Base[p]) \subset L_0\times L_{\ell}$ be a set of pairs of input/output ports.  
Each $(u_0,u_\ell)\in \Pairs(\Base[p])$
has the property that there exists a {\em unique} shortest path $(u_0,u_1,\ldots,u_\ell)$.
Moreover, $\lab(u_0,u_1) = \cdots = \lab(u_{\ell-1},u_{\ell})$, any two of these paths are edge disjoint, and the edge set $E$ is precisely the union of these paths over all pairs in $\Pairs(\Base[p])$.
\end{itemize}
These properties imply that the number of vertices and edges in $\Base[p]$ is 
$\nBase[p] \bydef (\ell+1)p$ and $\mBase[p] \bydef |E(\Base[p])| = \ell\cdot |\Pairs(\Base[p])|$.

Refer to~\cite{Alon01,AbboudB16} for constructions of $\Base[p]$ satisfying these requirements, 
or to \cite{CE06} for a construction without the layered structure.  For the sake of completeness we
give a short sketch of how $\Base[p]$ is constructed using average-free sets~\cite{Alon01,AbboudB16}.
Let $\Labels[p] \subset \{1,\ldots,\lfloor p/\ell \rfloor\}$ 
be an $\ell$-average-free set, i.e., one for which the equation
\begin{align*}
\ell \cdot x_0 &= x_1 + x_2 + \cdots + x_{\ell}, \; \; \mbox{where $x_0,x_1,\ldots,x_{\ell}\in \Labels[p]$}
\intertext{%
has no solutions, except the trivial $x_0 = x_1 = \cdots = x_{\ell}$.  
Let $u_{i,j}$ denote the $j$th vertex in $L_i$.  The edge set consists of 
}
E &= \{(u_{i,j}, u_{i+1,j'}) \;|\; i\in[0,\ell) \mbox{ and } (j'-j) \operatorname{mod} p \in \Labels[p]\},
\intertext{%
with $\lab(u_{i,j},u_{i+1,j'}) = (j'-j)\operatorname{mod} p$.  The pair set consists of
}
\Pairs(\Base[p]) &= \{(u_{0,j}, u_{\ell,(j+\ell x)\operatorname{mod} p}) \;|\; \mbox{ for all $j\in\{0,\ldots,p-1\}$ and $x\in\Labels[p]$}\}
\end{align*}
The $\ell$-average free property of $\Labels[p]$ ensures that $(u_{0,j}, u_{1,j+x},u_{2,j+2x},\ldots,u_{\ell, j+\ell x})$ is the unique
shortest path between its endpoints.

\subsection{The Second Base Graph}

Roughly speaking, $\DblBase[p]$ is obtained by taking a certain product of two copies of $\Base[\sqrt{p}]$.\footnote{Here we let $\sqrt{p}$ be short for $\floor{\sqrt{p}}$.  Ignoring issues of integrality only introduces $1+o(1)$ factors in all the bounds.}
Let $L_0^0\cup\cdots\cup L_\ell^0$ and $L_0^1\cup\cdots\cup L_\ell^1$ be the 
vertex sets of copies $\Base^0 [\sqrt p]$ and $\Base^1 [\sqrt p]$, each with respective pair-sets $\Pairs^0$ and $\Pairs^1$.
$\DblBase[p]$ is a layered graph with vertex set $\Layer{0} \cup\cdots\cup \Layer{2\ell}$ where
$\Layer{i} = L_{i/2}^0 \times L_{i/2}^1$ when $i$ is even and $\Layer{i} = L_{\ceil{i/2}}^0 \times L_{\floor{i/2}}^1$ when $i$ is odd.
Vertices in $\DblBase$ are identified with vertex pairs from $V(\Base^0)\times V(\Base^1)$.
When $i$ is even, an edge $((u,v),(u',v))$ exists between layers $\Layer{i}$ and $\Layer{i+1}$ iff $(u,u')\in E(\Base^0)$.  Similarly, 
when $i$ is odd, an edge $((u,v),(u,v'))$ exists between layers $\Layer{i}$ and $\Layer{i+1}$ iff $(v,v')\in E(\Base^1)$.
An edge in $\DblBase$ inherits the label of the corresponding edge in $\Base$, so the label set for $\DblBase[p]$ is
$\Labels[\sqrt{p}]$.
The pair-set for $\DblBase$ is defined to be 
\[
\Pairs(\DblBase[p]) = \{((u_0,v_0), (u_{\ell},v_{\ell})) \;|\; (u_0,u_\ell)\in \Pairs^0 \mbox{ and } (v_0,v_{\ell}) \in \Pairs^1\}.
\]  
Observe that any length-$2\ell$ path from layer $\Layer{0}$ to $\Layer{2\ell}$ corresponds to picking edges alternately
from two paths, one from $L_0^0$ to $L_\ell^0$ in $\Base^0$ and one from $L_0^1$ to $L_\ell^1$ in $\Base^1$.
Lemma~\ref{lem:properties-of-DblBase} summarizes the relevant properties of $\DblBase$ and $\Pairs(\DblBase)$.

\begin{lemma}\label{lem:properties-of-DblBase}
Let $\Loss_{\ell}(p)$ be a non-decreasing function of $p$ 
such that $|\Labels[p]| \ge p/\Loss_{\ell}(p)$, $|\Labels[p]| \le p/2$, and $|\Pairs(\Base[p])| \ge p^2 / \Loss_{\ell}(p)$.
The graph $\DblBase = \DblBase[p]$ has the following properties.
\begin{enumerate}
\item It has $\nDblBase[p] \le (2\ell+1)p$ vertices and $\mDblBase[p] \ge (1-o(1))2\ell p^{3/2} / \Loss_{\ell}(\sqrt{p})$ edges.
\item The vertices of each pair in $\Pairs(\DblBase[p])$ are connected by a unique shortest path in $\DblBase[p]$,
whose edge labels alternate between two labels in $\Labels[\sqrt{p}]$.
\item By definition, $|\Pairs(\DblBase[p])| = (|\Pairs(\Base[\sqrt{p}]|)^2 \ge p^2 / (\Loss_{\ell}(\sqrt{p}))^2$.
\end{enumerate}
\end{lemma}

\begin{proof}
Part 1. Each layer of $\DblBase$ contains $(\sqrt{p})^2$ vertices; there is no harm in adding dummy vertices 
to round it up to $p$.  There are at least $\ell \sqrt{p}^2/\Loss_{\ell}(\sqrt{p})$ edges in each of $\Base^0$ and $\Base^1$, 
and each edge of $\Base^0,\Base^1$ is duplicated $\sqrt{p}$ times in the construction of $\DblBase$.
Parts 2,3. Follows directly from the construction of $\DblBase$, and that $\Pairs(\Base[\sqrt{p}])$ has unique
shortest paths in $\Base$.
\end{proof}

A standard extension of Behrend's construction~\cite{Behrend46} of progression-free sets (see~\cite[Appendix]{AbboudB16})
shows that $\Loss_{\ell}(p) = 2^{O(\sqrt{\log p\log \ell})}$, so if $\ell = p^{o(1)}$ then $\Loss_{\ell}(p) = p^{o(1)}$ as well,
and if $\ell = p^{\epsilon}$ for an $\epsilon>0$ then $\Loss_{\ell}(p) = p^{\delta}$ for some $\delta=\delta(\epsilon) > \epsilon$.
We are most interested in the cases when $\ell, \Loss_{\ell}(p) = p^{o(1)}$.

\subsection{A Recursive Construction}

In this section we construct a hierarchy $\{(\Graph_k,\Pairs_k)\}_{k\ge 1}$ of hard graphs $(\Graph_k)$ and corresponding 
pair-sets $(\Pairs_k)$ such that each pair in $\Pairs_k$ has a unique shortest path in $\Graph_k$.
We will show that, 
for any $k \ge 2$ and sufficiently small constant $c_k$,  
any spanner of $\Graph_k$ with stretch function $f(d) = d + c_k d^{1-\frac{1}{k}} + \tilde{O}(1)$ must include at least $|\Pairs_k|$ edges.
Each $\Graph_k[p]$ is a layered graph with $p$ input ports, $p$ output ports, and some number of interior layers.
In other words, the first layer (``input ports'') and last layer (``output ports'') have size $p$ each while the interior layers may have different sizes, and each node pair in $\Pairs_k$ is composed of one input port and one output port.
Let $\RevGraph_k[p]$ denote the graph with the same topology as $\Graph_k[p]$ but with layers reversed; that is, the roles of input and output ports are swapped.

\paragraph{The Base Case.}
The base case graph $\Graph_1[p] = (\{1,\ldots,2p\}, \{1,\ldots,p\}\times\{p+1,\ldots,2p\})$ is a complete bipartite graph
on $2p$ vertices and its corresponding pair-set $\Pairs_1[p] = \{1,\ldots,p\}\times\{p+1,\ldots,2p\}$ has size $p^2$. 

\paragraph{The Inductive Case.}
Let us first give a very informal overview of the construction, then discuss how we plan to prove its correctness.
The goal is to produce a new graph $\Graph_k$ that contains within it many copies of $\Graph_{k-1}$.
The shortest path $P_{s,t}$ for each $(s,t)\in \Pairs_k$ joins an input port $s$ to an output port $t$, in $\Graph_k$,
and meanders through many copies of $\Graph_{k-1}$.  When $P_{s,t}$ goes through a copy of $\Graph_{k-1}$ it 
enters and exits it at a particular input/output port pair, say $(x,y)$.  We hope that $(x,y)\in \Pairs_{k-1}$ (a success); if this holds
for {\em all} the copies of $\Graph_{k-1}$ intersected by $P_{s,t}$ then any aggressive sparsification of these copies
will introduce a significant additive error in each copy.  Unfortunately, while $|\Pairs_{k-1}|$ is large, it is not {\em that} large.
Only a tiny $o(1)$-fraction of the set of input/output port pairs of $\Graph_{k-1}$ appear in $\Pairs_{k-1}$.  
Thus, if $P_{s,t}$ walks into and out of each $\Graph_{k-1}$
through random ports, it is likely to miss the pairs in $\Pairs_{k-1}$ (a failure).  

The problem with this approach is not the random assignment of input/output ports but the {\em independence} across
copies of $\Graph_{k-1}$.  We solve this problem by \emph{correlating} the success or failure events associated with $P_{s,t}$.
That is, we ensure that $P_{s,t}$ either enters/leaves \emph{every} copy of $\Graph_{k-1}$ along a pair in $\Pairs_{k-1}$, 
or it enters/leaves \emph{no} copy of $\Graph_{k-1}$ using a pair in $\Pairs_{k-1}$.
Thus, many of the potential pairs are useless and may be discarded, but some of the pairs $(s,t)$
must accumulate lots of error at each copy of $\Graph_{k-1}$ that $P_{s,t}$ touches.\\

We now give this argument in more formality.
When $k\ge 2$ we construct $\Graph_k[p]$ from $\Graph_{k-1}[\cdot]$ and $\DblBase = \DblBase[p]$ as follows.  
Let the label-set of $\DblBase$
be $\Labels = \Labels[\sqrt{p}]$ and $p' = |\Labels|$.
Let $\Graph_{k-1}, \RevGraph_{k-1}$ be the standard and reversed copies of 
$\Graph_{k-1}[p']$
and 
$\pi : \Labels \rightarrow \{1,\ldots,p'\}$ be a {\em port assignment permutation} selected uniformly at random.

Recall that $\DblBase$ consists of layers $\Layer{0},\ldots,\Layer{2\ell}$.  
Layers $\Layer{0}$ and $\Layer{2\ell}$ become the input and output ports of $\Graph_k$ and are left as-is.  
For each vertex $u$ in an interior layer $\Layer{i}$, we replace $u$ with a graph $\Graph(u)$, 
which is a copy of $\Graph_{k-1}$ if $i$ is odd and $\RevGraph_{k-1}$ if $i$ is even.
For each former edge $(u,u') \in \Layer{i}\times\Layer{i+1}$ in $\DblBase$ with $\lab(u,u')=a$, 
we replace it with a path of length $(2\ell-1)^{k-1}$ connecting the $\pi(a)$th output port of 
$\Graph(u)$ (or leave it at $u$ if $i=0$) and the $\pi(a)$th input port of $\Graph(u')$ (or leave it at $u'$ if $i+1=2\ell$.)
The resulting graph is $\Graph_k[p]$; see Figure~\ref{fig:LB-construction} for a diagram.  
It remains to define the new pair-set $\Pairs_k[p]$.

Let $(u_0,u_{2\ell}) \in \Pairs(\DblBase)$ be one of the pairs in $\DblBase$, and suppose
the edges on the unique shortest path from $u_0$ to $u_{2\ell}$ alternate between labels `$a$' and `$b$'.
The corresponding path $Q_{(u_0, u_{2\ell})}$ in $\Graph_k$ passes through some $\Graph(u_1),\Graph(u_2),\ldots,\Graph(u_{2\ell-1})$,
where $\Graph(u_1),\Graph(u_3),\ldots$ are copies of 
$\Graph_{k-1}[p']$ and $\Graph(u_2),\Graph(u_4),\ldots$ are copies of $\RevGraph_{k-1}[p']$. 

By construction, $Q_{(u_0, u_{2\ell})}$ enters $\Graph(u_i)$ at the $\pi(a)^{\mathrm{th}}$ input port and leaves at the $\pi(b)^{\mathrm{th}}$
output port, if $i$ is odd, or the reverse if $i$ is even.  Up to reversal, the input/output terminals through each $\Graph(u_i)$
are identical, for all $i \in [1,2\ell-1]$.  The pair-set $\Pairs_{k-1}[p]$ consists of all $(u_0,u_{2\ell})\in \Pairs(\DblBase)$ (whose
unique shortest path in $\DblBase$ is labeled with, say, $a,b$) for which $(\pi(a),\pi(b)) \in \Pairs_{k-1}[p']$.

\begin{figure}[h]
\centering
\scalebox{.45}{\includegraphics{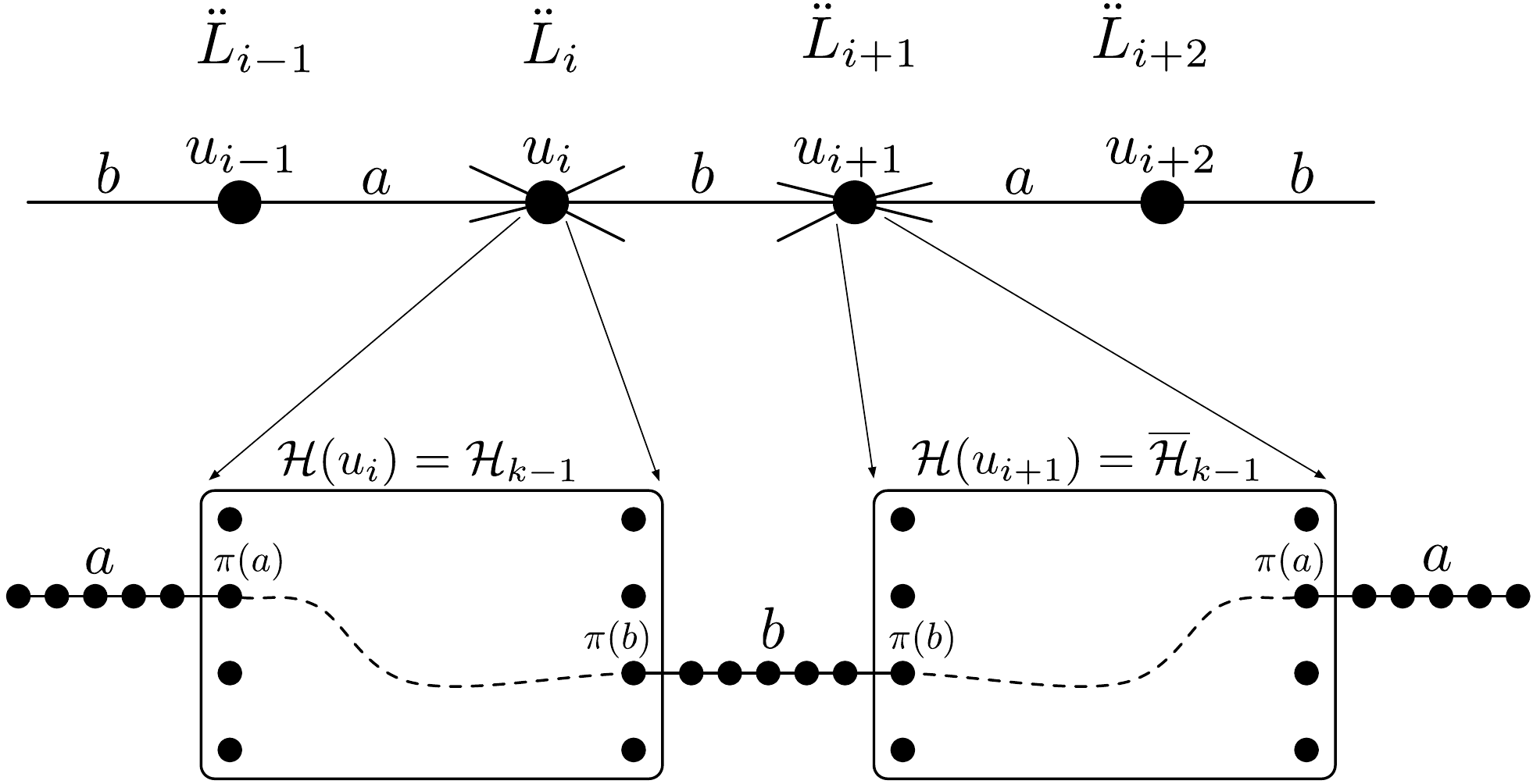}}
\caption{\label{fig:LB-construction}\small The edge-labels of a shortest path $(u_0,u_1,\ldots,u_{2\ell})$ for a pair $(u_0,u_{2\ell})\in\Pairs(\DblBase)$ always alternate between some $a,b\in\Labels$.
To form $\Graph_k$ we substitute for each vertex $u\in\DblBase$ a graph $\Graph(u)$, which is a 
copy of either $\Graph_{k-1}$
or $\RevGraph_{k-1}$, depending on whether $u$ appears in an odd or even numbered layer, respectively,
then replace each edge $(u,u')$ with a path of length $(2\ell-1)^{k-1}$.  The endpoints of this path
are the $(\pi(\lab(u,u')))$th output/input ports of $\Graph(u)$ and $\Graph(u')$.
If $(\pi(a),\pi(b))$ is in $\Pairs_{k-1}$, the pair set of $\Graph_{k-1}$, then there is a {\em unique}
shortest path in $\Graph_{k-1}$ from its $\pi(a)$th input port to its $\pi(b)$th output port.}
\end{figure}

\begin{lemma} \label{lem: pair set size}
The expected size of $\Pairs_k$ is $\f{p^2}{(\Loss_{\ell}(\sqrt{p}))^2} \cdot \f{|\Pairs_{k-1}|}{(p')^2}$.  
Assuming $\Loss_{\ell}(\cdot)$ is a nondecreasing function for all $\ell$, the expected
size of $\Pairs_k$ is on the order of $p^2 / (\Loss_{\ell}(\sqrt{p}))^{2(k-1)}$.
\end{lemma}

\begin{proof}
By definition of $\Loss_{\ell}(\cdot)$ and the construction of $\DblBase[p]$, 
there are $p^2 / (\Loss_{\ell}(\sqrt{p}))^2$ candidate pairs in $\Pairs(\DblBase[p])$,
each of which, say $(u_0,u_{2\ell})$, is associated with two alternating labels, say $a$ and $b$.
Since a uniformly random
input/output pair $(\pi(a),\pi(b))$ is in $\Pairs_{k-1}$ with probability $|\Pairs_{k-1}|/(p')^2$,
$(u_0,u_{2\ell})$ is retained in $\Pairs_k$ with exactly this probability.

The second part of the proof is by induction.  Let $\Loss = \Loss_{\ell}(\sqrt{p})$.  
When $k=1$ we have $|\Pairs_1| = p^2 = p^2 / \Loss^0$.
By the inductive hypothesis and the fact that $\Loss_{\ell}$ is non-decreasing, 
$|\Pairs_{k-1}| \ge (p')^2 / \Loss^{2(k-2)}$ (in expectation)
and so $|\Pairs_k| \ge (p^2 / \Loss^2) \cdot (1/\Loss^{2(k-2)}) = p^2 / \Loss^{2(k-1)}$ (in expectation). 
\end{proof}

\begin{lemma} \label{lem: unique paths}
If $(u_0,u_{2\ell}) \in \Pairs_k$ then there is a unique shortest path from $u_0$ to $u_{2\ell}$ in $\Graph_k$,
which passes through exactly $(2\ell-1)^{k-1}$ complete bipartite graphs (copies of $\Graph_1[\cdot ]$)
and has $\dist_{\Graph_k}(u_0,u_{2\ell}) = (2(k-1)\ell + 1)(2\ell-1)^{k-1}$, 
which is exactly the distance from the first to last layer of $\Graph_k$.
\end{lemma}

\begin{proof}
We prove the claim by induction.  
Let $d_k$ be the length of shortest paths for $\Pairs_k$ in $\Graph_k$.
Uniqueness of shortest paths in $\Graph_1$ is obvious, and $d_1$ is clearly 1.
In the construction of $\Graph_k$, shortest paths for pairs in $\Pairs(\DblBase[p])$
traverse $2\ell$ edges via $2\ell-1$ interior layers.  Each edge becomes a path of length $(2\ell-1)^{k-1}$
and each interior vertex becomes a copy of $\Graph_{k-1}$.
Suppose a pair $(u_0,u_{2\ell})\in\Pairs(\DblBase[p])$
is preserved in $\Pairs_k$ and let $(u_0,u_1,\ldots,u_{2\ell})$ be the shortest path in $\DblBase$.
By the definition of $\Pairs_k$ and the inductive hypothesis, there are unique shortest
paths from the given input port to the given output port in $\Graph(u_1),\Graph(u_2),\ldots,\Graph(u_{2\ell-1})$.
Each passes through $(2\ell-1)^{k-2}$ complete bipartite graphs and each has length exactly $d_{k-1}$.  
Any alternative shortest path would have to visit consecutive layers without backing up to earlier layers,
and would therefore visit $u_0, \Graph(u_1'),\ldots,\Graph(u_{2\ell-1}'), u_{2\ell}$, 
for some $(u_1',\ldots,u_{2\ell-1}')\neq (u_1,\ldots,u_{2\ell-1})$.  
This, however, violates Lemma~\ref{lem:properties-of-DblBase}(2) on the uniqueness of shortest
paths between pairs in $\Pairs(\DblBase)$.  We have the recurrence
\begin{align*}
d_1 &= 1\\
d_k &= (2\ell-1)d_{k-1}   +   2\ell(2\ell-1)^{k-1},
\end{align*}
which has the closed form $d_k = (2(k-1)\ell + 1)(2\ell-1)^{k-1}$.
\end{proof}

\begin{definition} \label{def: critical}
An edge $e$ is \emph{critical} for a pair $(u_0,u_{2\ell})\in \Pairs_k$ if it lies in a copy of $\Graph_1$ (a complete bipartite graph), and it is on the unique shortest $u_0$--$u_{2\ell}$ path.
\end{definition}

\begin{lemma}\label{lem:removing-critical-edges}
Let $\Graph_k'$ be $\Graph_k$, with all critical edges for $(u_0,u_{2\ell})$ removed.
Then $\dist_{\Graph_k'}(u_0,u_{2\ell}) \ge \dist_{\Graph_k}(u_0,u_{2\ell}) + 2(2\ell-1)^{k-1}$.
\end{lemma}

\begin{proof}
Let $(u_0,u_1,\ldots,u_{2\ell})$ be the unique shortest path from $u_0$ to $u_{2\ell}$ in $\DblBase$.
In $\Graph_k'$, if we take a path that does {\em not} pass through $\Graph(u_1),\ldots,\Graph(u_{2\ell-1})$
then at some point it must move back to an earlier layer (from, say, layer $i$ to layer $i-1$) before advancing
forward again (from layer $i-1$ to $i$, and onward to $2\ell$).  Since each edge in $\DblBase$ corresponds
to a path of length $(2\ell-1)^{k-1}$ in $\Graph_k$, such a detour increases the path length by at least $2(2\ell-1)^{k-1}$.
On the other hand, if we do take a path passing through $\Graph(u_1),\ldots,\Graph(u_{2\ell-1})$ then
it must use the same input/output ports as the unique shortest $u_0$--$u_{2\ell}$ path.  By the inductive hypothesis,
the additive stretch inside each of these subgraphs must be $2(2\ell-1)^{k-2}$.  (This is true when $k=2$ as well,
since in this case $\Graph(u_1),\ldots,\Graph(u_{2\ell-1})$ are complete bipartite graphs, and removing the critical edge
increases the distance from 1 to 3 in each one.)  Thus, the total additive stretch is at least 
$(2\ell-1)\cdot 2(2\ell-1)^{k-2} = 2(2\ell-1)^{k-1}$.
\end{proof}

\begin{lemma} \label{lem:pair-distances}
The shortest paths in $\Graph_k$ for two pairs in $\Pairs_k$ share no critical edges.  As a consequence,
any spanner of $\Graph_k$ with fewer than $|\Pairs_k|$ edges must stretch some pair of vertices at distance
$d = (2(k-1)\ell+1)(2\ell-1)^{k-1}$ by an additive $2(2\ell-1)^{k-1} \approx \f{2}{(k-1)^{1-1/k}}\cdot d^{1-\f{1}{k}}$.
\end{lemma}

\begin{proof}
The proof is by induction on $k$; it clearly holds when $k=1$.
For $k\ge 2$, each pair $(u_0,u_{2\ell}) \in \Pairs_k$ is identified with a pair of labels $a,b \in \Labels$, 
which determines the input/output ports of $\Graph(u_1),\ldots,\Graph(u_{2\ell-1})$ 
used in the shortest $u_0$--$u_{2\ell}$ path. 
No other pair $(u_0',u_{2\ell}')$ whose shortest path intersects some $\Graph(u_i)$ 
can be associated with the same two labels, hence it must enter and exit $\Graph(u_i)$ at different
input/output ports than $(u_0,u_{2\ell})$.  By the inductive hypothesis, $(u_0,u_{2\ell})$ and $(u_0',u_{2\ell}')$
share no critical edges in any $\Graph(u_i)$, and therefore no critical edges in $\Graph_k$.
\end{proof}

We now analyze the size and density of $\Graph_k[p]$.  Let $n_k[p]$ and $m_k[p]$ be the number of vertices and edges in $\Graph_k[p]$.  The construction of $\Graph_k$ gives the following recursive definition.
\begin{align*}
\nBase[p]		&= (\ell+1)p		& \mBase[p] 		&\ge \frac{\ell p^2}{\Loss_{\ell}(p)}\\
\nDblBase[p]	&= (2\ell+1)p		& \mDblBase[p]		&\ge \frac{2\ell p^{3/2}}{\Loss_{\ell}(\sqrt{p})},  \;\; \mDblBase[p] \le 2\ell p^{3/2}\\
n_1[p]		&= 2p			& m_1[p]			&= p^2\\
n_k[p]		&= 2p + (2\ell-1)p\cdot n_{k-1}\left[|\Labels[\sqrt{p}]|\right] + (2\ell-1)^{k-1} \cdot \mDblBase[p] \;\;\;\;\;\;\;
							& m_k[p]			&\ge (2\ell-1)p\cdot m_{k-1}\left[|\Labels[\sqrt{p}]|\right]\\
\end{align*}

\begin{lemma} \label{lem: lb calculations}
For all $k \ge 1$ and sufficiently large $p$, 
$n_k[p] \le c(2\ell)^{k} p^{2-\frac{1}{2^{k-1}}}$
and 
$m_k[p] \ge (2\ell-1)^{k-1} p^2 / \Loss^{2(k-1)}$, where $\Loss=\Loss_{\ell}(\sqrt{p})$.
\end{lemma}

\begin{proof}
The bounds clearly hold when $k=1$.  Assuming the claim holds inductively for $n_{k-1}$, we have
\begin{align*}
n_k[p] &\leq (2\ell)p\cdot n_{k-1}\left[\sqrt{p}/2\right] + (2\ell-1)^{k-1} \mDblBase[p]		& (|\Labels[\sqrt{p}]| \le \sqrt{p}/2)\\
	&\leq (2\ell)p\cdot \left[c(2\ell)^{k-1} (\sqrt{p}/2)^{2-\frac{1}{2^{k-2}}}  \right]	+ (2\ell-1)^{k-1} 2\ell p^{3/2}\\ 
	&< (c/2^{2-\frac{1}{2^{k-2}}})(2\ell)^{k} p^{2-\frac{1}{2^{k-1}}} + (2\ell)^{k} p^{3/2}\\		
	&< c(2\ell)^{k}p^{2-\frac{1}{2^{k-1}}}									& (\mbox{for, say, $c=2$)}\\
\intertext{and assuming the claim holds inductively for $m_{k-1}$ we have,}
m_k[p] &\ge (2\ell-1)p\cdot m_{k-1}\left[\sqrt{p} / \Loss \right] 						& (|\Labels[\sqrt{p}]| \ge \sqrt{p}/\Loss_{\ell}(\sqrt{p}))\\
	&\ge (2\ell-1)p \cdot (2\ell-1)^{k-2} (\sqrt{p} / \Loss)^2 / \Loss^{2(k-2)}					& \mbox{($\Loss_{\ell}(\cdot)$ nondecreasing)}\\
	&\ge (2\ell-1)^{k-1} p^2 / \Loss^{2(k-1)}.							
\end{align*}
We are mainly interested in cases in which $\ell$ is not too large, i.e. $\ell,\Loss = (n_k[p])^{o(1)}$.
In this case the density
of $\Graph_k[p]$ is $m_k[p]/n_k[p] \ge (n_k[p])^{\f{1}{2^{k}-1} - o(1)}$.
\end{proof}

\begin{theorem}\label{thm:lb-sublinearadditive} {\bf (Sublinear Additive Spanner Lower Bounds)}
For any integer $k\ge 2$ and a sufficiently small constant $c_k = O(1/k)$, 
any spanner construction with stretch function bounded 
by $f(d) \le d+c_k d^{1-\f{1}{k}} + \tilde{O}(1)$ has size 
$\Omega(n^{1+\frac{1}{2^{k+1}-1}-o(1)})$ in the worst case.
\end{theorem}

\begin{proof}
Let $\Graph_k[p]$ be the input graph with respect to some sufficiently large $\ell = (\log p)^{O(1)}$. 
For this parameterization $\Loss_{\ell}(p) = 2^{O(\sqrt{\log p\log\log p})} = p^{o(1)}$, we have that $m_k[p]$ and the size of the pair-set $\Pairs_k$
are both $n^{1+\frac{1}{2^{k}-1}-o(1)}$, where $n=n_k[p]$.   Any spanner with size less than $|\Pairs_k|$
must stretch some pair at distance $d_k = (2(k-1)\ell + 1)(2\ell-1)^{k-1}$ to $d_k + 2(2\ell-1)^{k-1}$, which is strictly
greater than $d_k + c_k d^{1-\f{1}{k}} + \tilde{O}(1)$ when $c_k < 2/(k-1)^{1-1/k}$ is sufficiently small 
and $\ell$ sufficiently large to make the
$\tilde{O}(1)$ error comparatively negligible.
\end{proof}

\begin{remark}
Since the diameter of $\Graph_k[p]$ is $O(d_k)$, any emulator
for $\Graph_k[p]$ {\em on the same vertex set} (i.e., without Steiner points) 
can be converted to a spanner with at most an $O(d_k) = n^{o(1)}$ blowup in the number of edges.
Thus, Theorem~\ref{thm:lb-sublinearadditive} applies to this class of emulators.
The argument breaks down for (Steiner) emulators since we can
preserve all distances for pairs in $\Pairs_k$ with just $O(n)$ edges, simply by replacing
all bipartite cliques (copies of $\Graph_1$) with stars.  See Theorem~\ref{thm:lb-DO}
for a lower bound that applies to emulators with Steiner points.
\end{remark}

\begin{theorem}\label{thm:lb-onepluseps} {\bf ($(1+\epsilon,\beta)$-Spanner Lower Bounds)}
Any $(1+\epsilon,\beta)$ spanner construction with worst-case size at most $n^{1+\f{1}{2^{k+1}-1} - \delta}$, $\delta>0$,
has $\beta = \Omega \paren{\f{1}{\epsilon(k-1)}}^{k-1}$.
\end{theorem}

\begin{proof}
Let $\Graph_k[p]$ be the input graph with respect to an $\ell \approx \f{1}{\epsilon}$ to be chosen shortly.
Any spanner with size $n^{1+\f{1}{2^{k+1}-1} - \delta} < |\Pairs_k|$ stretches 
a pair at distance
$d_k = (2(k-1)\ell + 1)(2\ell-1)^{k-1}$ to 
$d_k + 2(2\ell-1)^{k-1} = d_k(1+\f{2}{2(k-1)\ell+1})$.
We choose $\ell\ge 2$ to be  minimal such that 
$\f{1}{2(k-1)\ell+1} \le \epsilon$, that is, $\ell = \ceil{\f{1-\epsilon}{2\epsilon(k-1)}}$
and the additive stretch is roughly $2\epsilon d_k$.  
In order for this to be a $(1+\epsilon,\beta)$-spanner
we would need 
\[
\beta = \Omega(\epsilon d_k) = \Omega((2\ell-1)^{k-1}) 
					= \Omega\paren{\paren{2\ceil{\fr{1-\epsilon}{2\epsilon(k-1)}} - 1}^{k-1}}.
\]
\end{proof}

Theorem~\ref{thm:lb-onepluseps} shows that the existing $(1+\epsilon, O(k/\epsilon)^{k-1})$-spanners
with size $O((k/\epsilon)^{O(1)}n^{1+\f{1}{2^{k+1}-1}})$ are optimal in the following sense.  If $k$ is constant
then we cannot improve $\beta$ by more than a constant factor $\approx (k^2)^{k-1}$ without
increasing the exponent to $1+\f{1}{2^k-1}-o(1)$.  Moreover, any constant reduction in the
exponent increases $\beta$ to $\Theta(1/(k\epsilon))^k$.  Once again, the argument of 
Theorem~\ref{thm:lb-onepluseps} applies to $(1+\epsilon,\beta)$-emulators that do not use Steiner points.

\begin{theorem}\label{thm:lb-DO} {\bf (Distance Oracle/Emulator Lower Bounds)}
Consider any data structure for the class of $n$-vertex undirected graphs
that answers approximate distance queries.
If its stretch function is:
\begin{itemize}
\item $f(d) \le d + c_k d^{1-\f{1}{k}} + \tilde{O}(1)$ for an integer $k$ and a sufficiently small constant $c_k < 2/(k-1)^{1-1/k}$, or
\item $f(d) \le (1 + \eps)d + \beta$ where $\beta = o\paren{\paren{\f{1}{\epsilon(k-1)}}^{k-1}}$
\end{itemize}
then on some graph, the data structure occupies at least $n^{1+\f{1}{2^k-1} -o(1)}$ bits of space.
\end{theorem}

\begin{proof}
The following proof strategy was employed by \Althofer{} et al.~\cite{Althofer+93} to bound the 
size of emulators.  It was also used by \Matousek~\cite{Matousek96} for bounding low-distortion embeddings into 
$l_{\infty}^{d}$, and by Thorup and Zwick~\cite{TZ05} and Abboud and Bodwin~\cite{AbboudB16} to 
bound the size of approximate distance data structures.

Fix a graph $\Graph_k = \Graph_k[p]$ with pair set $\Pairs_k$.  For any subset $\Pairs'\subseteq \Pairs_k$
let $G(\Pairs')$ be obtained from $\Graph_k$ be removing {\em all} the critical edges for each pair in $\Pairs'$,
and let $\mathcal{G} = \{G(\Pairs') \;|\; \Pairs' \subseteq \Pairs_k\}$ be the class of $2^{|\Pairs_k|}$ such graphs.
Fix any two graphs $G_A,G_B \in \mathcal{G}$.  There must exist {\em some} pair $(u,v)\in \Pairs_k$ such 
that $G_A$ contains all of the critical edges for $(u,v)$ whereas $G_B$ contains none of them.
By Lemma~\ref{lem:removing-critical-edges} we have 
\begin{align*}
\dist_{G_A}(u,v) &= d_A \bydef (2(k-1)\ell + 1)(2\ell-1)^{k-1}\\
\dist_{G_B}(u,v) &\geq d_B \bydef (2(k-1)\ell + 1)(2\ell-1)^{k-1} + 2(2\ell-1)^{k-1}
\end{align*}
If $f(d_A) < d_B$ then no single data structure (bit string) can be used to encode two distinct graphs $G_A,G_B\in \mathcal{G}$,
implying that the data structures for $\mathcal{G}$ must use at least $\log_2(2^{|\Pairs_k|})$ bits, on average.
If $\ell = (\log p)^{O(1)} = (\log n)^{O(1)}$ is sufficiently large and $c_k$ sufficiently small then
\[
f(d_A) \;=\; d_A + c_k d_A^{1-\f{1}{k}} + \tilde{O}(1) < d_B \;\approx\; d_A + \fr{2}{(k-1)^{1-1/k}} \cdot d_A^{1-\f{1}{k}}.
\]
For these parameters $|\Pairs_k| = n^{1+\f{1}{2^k-1} -o(1)}$, where the $n^{-o(1)}$ factor is $2^{-O(\sqrt{\log n\log\log n})}$.
We extend this argument to the case of $(1+\eps, \beta)$ type stretch using an identical argument to the one given in Theorem \ref{thm:lb-onepluseps}.
\end{proof}

\section{New Upper Bounds on $(1+\epsilon,\beta)$-Spanners}\label{sect:new-upper-bounds}

Thorup and Zwick~\cite{TZ06} gave a very simple randomized 
construction of an emulator with size $O(kn^{1-\frac{1}{2^{k+1}-1}})$ 
and stretch function $f(d) = d + O(kd^{1-1/k} + 3^k)$.  Alternatively, one can view
this as a $(1+\epsilon, O(k/\epsilon)^{k-1})$-emulator for {\em every} $\epsilon>0$,
where the optimal choice of $\epsilon$, as a function of $d$, is $\epsilon = \Theta(k/d^{1/k})$.

\subsection{The Thorup-Zwick Emulator}

The Thorup-Zwick emulator is parameterized by an integer $k\ge 2$.  Let $G=(V,E)$ be the input graph.
One samples vertex sets $V = V_0 \supset V_1 \supset V_2 \cdots \supset V_{k}$
where vertices in $V_i$ are promoted to $V_{i+1}$ with probability $q_{i+1}/q_i$, 
so $\E[|V_i|] = q_i n$.  
Define $p_i(u)$ to be the closest $V_i$-vertex to $u$, breaking ties in a consistent manner.
Define $\Ball(u,r) = \{v \;|\; \dist(u,v) \le r\}$ to be the set of vertices inside the radius-$r$ ball centered at $u$
and let $\Ball_i(u)$ be short for $\Ball(u, \dist(u,p_i(u))-1)$.
For $i\geq k+1$, $p_i(u)$ does not exist and $\Ball_i(u)$ is the entire graph, by definition.
The emulator edge set consists of $E_0 \cup E_1 \cup \cdots \cup E_{k}$,
where $E_i$ is defined as follows.
\[
E_i = \left\{(u,v) \;|\; u,v\in V_i \mbox{ and } v\in \Ball_{i+1}^{\ }(u)\right\}  \;\cup\; \left\{(u,p_{i+1}^{\ }(u)) \;|\; u\in V\right\}.
\]
The length of all emulator edges is precisely the distance between their endpoints in $G$.  Since $|\Ball_{i+1}(u)|$ 
is $q_{i+1}^{-1}$ in expectation, the expected number of edges contributed by $E_i$ is $n + nq_i^2/q_{i+1}$, 
for $i < k$, and is $(nq_{k})^2$ when $i=k$.  Setting $q_i = n^{-\f{2^i-1}{2^{k+1}-1}}$ makes
the size of the emulator $O(kn^{1+\f{1}{2^{k+1}-1}})$ in expectation.
In order to obtain a $d + O(kd^{1-\f{1}{k}})$-type 
stretch bound for all distances 
$d \le D$, it actually suffices to restrict $E_i$ to pairs at distance at most 
$(r+2)^i$, where $r = D^{1/k}$.
Letting $P(u,v)$ be any shortest path from $u$ to $v$, the subgraph 
$S_{TZ}(k,r) = (V, E_0'\cup E_1'\cup\cdots\cup E_{k}')$ is 
a spanner, where
\[
E_i' = \bigcup_{\substack{(u,v)\in E_i \; :\\ v\in \Ball(u,(r+2)^i)}} P(u,v).
\]
As we show in Sections~\ref{sect:even-sparser} and \ref{sect:stretch-analysis}, 
the spanner $S_{TZ}(k,r)$ behaves exactly like the emulator for all distances up to $D$,
i.e., it has stretch function $d+O(kd^{1-\f{1}{k}})$ for all sufficiently large $d\le D$.\footnote{Thorup and Zwick~\cite[p. 809]{TZ06}
also noted that their emulator can be converted to a spanner, but their sketch of how to do this was incorrect.}
However, choosing the optimum sampling probabilities as a function of $r,k,n$ is no longer trivial.  
Since each path in $E_i$ contributes $(r+2)^i$ edges, the spanner size is on the order of 
$kn$ (for paths of the form $P(u,p_i(u))$) plus
\[
\frac{n}{q_1} + \frac{nq_1^2r}{q_2} + \frac{nq_2^2r^2}{q_3} + \cdots \frac{nq_{k-1}^2r^{k-1}}{q_{k}} + (nq_{k})^2r^{k}.
\]
Assuming this sum is minimized when $E_0',E_1',\ldots,E_{k}'$ contribute equally, 
we have the following equalities:
\begin{align*}
q_2 &= rq_1^3								& \mbox{(balancing $E_0'$ and $E_1'$)}\\
q_3 &= r^2 q_2^2 q_1 						& \mbox{(balancing $E_0'$ and $E_2'$)}\\
&\vdots\\
q_{k} &= r^{k-1} q_{k-1}^2q_1 					& \mbox{(balancing $E_0'$ and $E_{k-1}'$)}\\
\intertext{If $q_i$ is constrained to be of the form $n^{-g(i)} r^{-h(i)}$, these equalities are satisfied when}
g(i) 				&= 2g(i-1) + g(1)			& \mbox{(for $i\ge 2$)}\\
				&= (2^i-1)g(1)				& \mbox{(by induction)}\\
\mbox{and } \; h(i)	&= 2h(i-1) + h(1) - (i-1)		& \mbox{(for $i\ge 2$)}\\
				&= (2^i-1)h(1) - 2^i + (i+1)		& \mbox{(by induction)}\\
\intertext{So $q_{k} = n^{-(2^{k}-1)g(1)} r^{-(2^{k}-1)h(1) + 2^{k} - (k+1)}$.  Plugging this equality
into the expression for $|E_{k}'|$ and balancing with $|E_0'|$, we have,
\[
|E_{k}'| \,=\, (nq_{k})^2 r^{k} \,=\, n^{2-2(2^{k}-1)g(1)} r^{-2[(2^{k}-1)h(1) - 2^{k} + (k+1)] + k} \,=\, n^{1+g(1)} r^{h(1)}  \,=\, |E_0'|, 
\]
which is minimized when
}
g(1) &= \f{1}{2^{k+1}-1}\\
h(1) &= \f{2^{k+1} - (k+2)}{2^{k+1}-1}.
\end{align*}
For example, when $k=2$ and $h(1)=4/7$ this leads to a $d+O(\sqrt{d})$-spanner for distances 
$d \le D = r^2$ having size
$O(r^{4/7}n^{8/7}) = O(D^{2/7}n^{8/7})$.  
Since $h(1)$ is strictly less than 1 for any fixed $k$, the spanner size is always 
$o(rkn^{1+\f{1}{2^{k+1}-1}}) = o(D^{1/k}kn^{1+\f{1}{2^{k+1}-1}})$.

\subsection{Even Sparser $(1+\epsilon,\beta)$-Spanners}\label{sect:even-sparser}

In order to form an even sparser spanner we substitute for $E_1'$ a subgraph whose size has no dependence on $r$ but
preserves the relevant distances well enough, up to an additive +2 error.  The following theorem is proved using 
the same path-buying algorithm for constructing additive 6-spanners~\cite{BKMP10,Pettie-Span09}.
The algorithm begins with the subgraph $E_0'$ and supplements it with an $\tilde{E}_1$ to guarantee +2 
stretch for each $u,v\in V_1$ that were connected by an edge in $E_1$.

\begin{theorem}\label{thm:path-buying} (see~\cite{BKMP10} and \cite{Pettie-Span09})
Suppose $V_1,V_2$ are sampled with probability $q_1$ and $q_2$, with $q_2 < q_1$. 
Then there is an edge-set $\tilde{E}_1$ with expected size $O(nq_1^2/q_2)$ such that if
$u,v\in V_1$ and $v \in \Ball_2(u)$, then
\[
\dist_{E_0' \cup \tilde{E}_1}(u,v) \le \dist(u,v) + 2.
\]
\end{theorem}

\begin{proof} (sketch)
We assume the reader is familiar with the path-buying algorithm and its analysis~\cite{BKMP10}.
Let $\mathcal{P} \subset {V_1 \choose 2}$ be the pairs for which we are guaranteeing good stretch,
i.e., $\{u,v\} \in \mathcal{P}$ if $u\in \Ball_2(v)$ or $v\in \Ball_2(u)$.  Since $|\Ball_2(u)|$ is $1/q_2$ in expectation, 
$|\mathcal{P}|$ is $O(nq_1^2/q_2)$ in expectation.  For each $\{u,v\} \in \mathcal{P}$ we evaluate $P(u,v)$ and buy
it (set $\tilde{E}_1 \leftarrow \tilde{E}_1 \cup P(u,v)$) 
if its current value exceeds its cost.  The value is the number of pairs $\{x,y\} \in \mathcal{P}$
with $x,y$ adjacent to $P(u,v)$ for which $\dist_{E_0'\cup \tilde{E}_1 \cup P(u,v)}(x,y) < \dist_{E_0'\cup \tilde{E}_1}(u,v)$.
It is argued by the pigeonhole principle that any path $P(u,v)$ not bought has $\dist_{E_0' \cup \tilde{E}_1}(u,v) \le \dist(u,v)+2$,
and that each pair in $\mathcal{P}$ is charged for $O(1)$ edges in $\tilde{E}_1$.
\end{proof}

We sample vertex sets $V = V_0 \supset V_1 \supset \cdots \supset V_{k}$ 
as before and construct the spanner $S(k,r)$ with edge set 
$E_0' \cup \tilde{E}_1 \cup E_2' \cup \cdots \cup E_{k}'$,
where $\tilde{E}_1$ is the edge set from Theorem~\ref{thm:path-buying}.
The expected size of the entire spanner is therefore
\[
\frac{n}{q_1} + \frac{nq_1^2}{q_2} + \frac{nq_2^2r^2}{q_3} + \cdots + \frac{nq_{k-1}^2r^{k-1}}{q_{k}} + (nq_{k})^2r^{k}.
\]
Letting $q_i = n^{-\f{2^i-1}{2^{k+1}-1}} r^{-h(i)}$, we balance the contribution of $E_0',\tilde{E}_1,E_2',\ldots,E_{k}'$
by having $h$ satisfy the following.
\begin{align*}
					h(2) &= 3h(1)    					& \mbox{(balancing $\tilde{E}_1$ and $E_0'$)}\\
\mbox{and for $i\ge 3$, }\, h(i) &= 2h(i-1) + h(1) - (i-1)			& \mbox{(balancing $E_{i-1}'$ and $E_0'$)}\\
			     &= (2^i-1)h(1) - 3\cdot 2^{i-2} + (i+1)			& \mbox{(by induction)}
\end{align*}
Following similar calculations, it follows that the spanner size is minimized when 
\[
h(1)= \f{3\cdot 2^{k-1} - (k+2)}{2^{k+1}-1}.
\]
We shall prove shortly that this spanner is, indeed, a $d + O(kd^{1-\f{1}{k}} + 3^k)$-spanner.  
For example, when $k=1$ we have $h(1)=2/7$, 
so it is a $d+O(\sqrt{d})$-spanner for all $d\le D \le r^2$ with size $O(r^{2/7}n^{8/7}) = O(D^{1/7}n^{8/7})$.  
For any fixed $k$, $h(1) < 3/4$, so the spanner has size $o(D^{\frac{3}{4k}} kn^{1+\f{1}{2^{k+1}-1}})$.

\begin{remark}
We were able to substitute $\tilde{E}_1$ for $E_1'$ without disturbing the exponent $1+\f{1}{2^{k+1}-1}$ of the spanner,
but only because the path-buying algorithm buys $O(nq_1^2/q_2)$ additional edges when initialized with the edge set $E_0'$.
In general we can use~\cite[Thm.~4.2]{Pettie-Span09} to substitute an $\tilde{E}_i$ for $E_i'$, but its size
is $O(n\sqrt{q_i/q_{i+1}})$.  This improves the exponent attached to $r$ but worsens the exponent attached to $n$.
For example, balancing $E_2'$ and $E_0'$ lets us put $q_3 = (q_1)^7r^{O(1)}$,
whereas balancing $\tilde{E}_2$ and $E_0'$ forces $q_3 = (q_1)^5$.
\end{remark}

\subsection{Stretch Analysis}\label{sect:stretch-analysis}

We analyze the stretch of the spanner $S=S(k,r)$ with edge set $E_0' \cup \tilde{E}_1 \cup E_2' \cup \cdots \cup E_{k}'$.
We will first consider two vertices $u,v$ at distance at most $\ell^i$, for some integers $\ell\ge 2, i\ge 0$.
We will assume for the time being that $r=\infty$ and calculate specific quantities related to the spanner distance $\dist_S(u,v)$
without considering the constraints imposed by a finite $r$.  Once these quantities are calculated,
it will be clear that the analysis goes through, so long as $\ell \le r$.
The pair $u,v$ can be either {\em complete} or {\em incomplete} (or both), as explained in the following definition.

\begin{definition}\label{def:succfail}
Define $\{\Succ{\ell}{i}, \, \Fail{\ell}{i}\}_{\ell \in [2,r], i\ge 0}$ to be integers such that for all $u,v$ with 
$\dist(u,v) \le \ell^i$, at least one of the following inequalities holds.  Here $S=S(k,r)$ is the spanner.
\begin{align*}
\dist_{S}(u,v) &\le \dist(u,v) + \Succ{\ell}{i}				& \mbox{(``$u\cdots v$ is {\em complete}'')}\\
\dist_{S}(u, p_{i+1}(u)) &\le \Fail{\ell}{i}				& \mbox{(``$u\cdots v$ is {\em incomplete}'')}
\end{align*}
\end{definition}

\begin{lemma}\label{lem:succfail-recursive}
The following values for $\{\Succ{\ell}{i}, \, \Fail{\ell}{i}\}_{\ell\in [2,r], i\ge 0}$ satisfy 
Definition~\ref{def:succfail}.
\begin{align*}
\Succ{\ell}{0} &= 0			& \mbox{ for all $\ell$}\\										
\Fail{\ell}{0} &= 1			& \mbox{ for all $\ell$}\\												
\Succ{\ell}{1} &= 6			& \mbox{ for all $\ell$}\\
\Fail{\ell}{1} &= \ell+3			& \mbox{ for all $\ell$}\\
\Succ{\ell}{i} &= \min\left\{\begin{array}{l}
\ell\cdot \Succ{\ell}{i-1}\\
(\ell-1)\cdot \Succ{\ell}{i-1} + 4\cdot \Fail{\ell}{i-1}	
\end{array}\right. 									& \mbox{ for all $\ell$ and $i\geq 2$}\\
\Fail{\ell}{i} &= \ell^i + 3\cdot \Fail{\ell}{i-1}					& \mbox{ for all $\ell$ and $i\ge 2$}
\end{align*}
\end{lemma}

\begin{proof}
In the base case $(i=0)$, we have $\ell^0 = 1$, so $u$ and $v$ are adjacent in the input graph.
If $(u,v)\in E_0'$ then $\dist_{S}(u,v)=1$ and if $(u,v)\not\in E_0'$ then it must be that $\dist(u,p_1(u))=1$,
so $\Succ{\ell}{0}=0, \Fail{\ell}{0}=1$ satisfy Definition~\ref{def:succfail} for all $\ell$.

\begin{figure}
\centering
\scalebox{.5}{\includegraphics{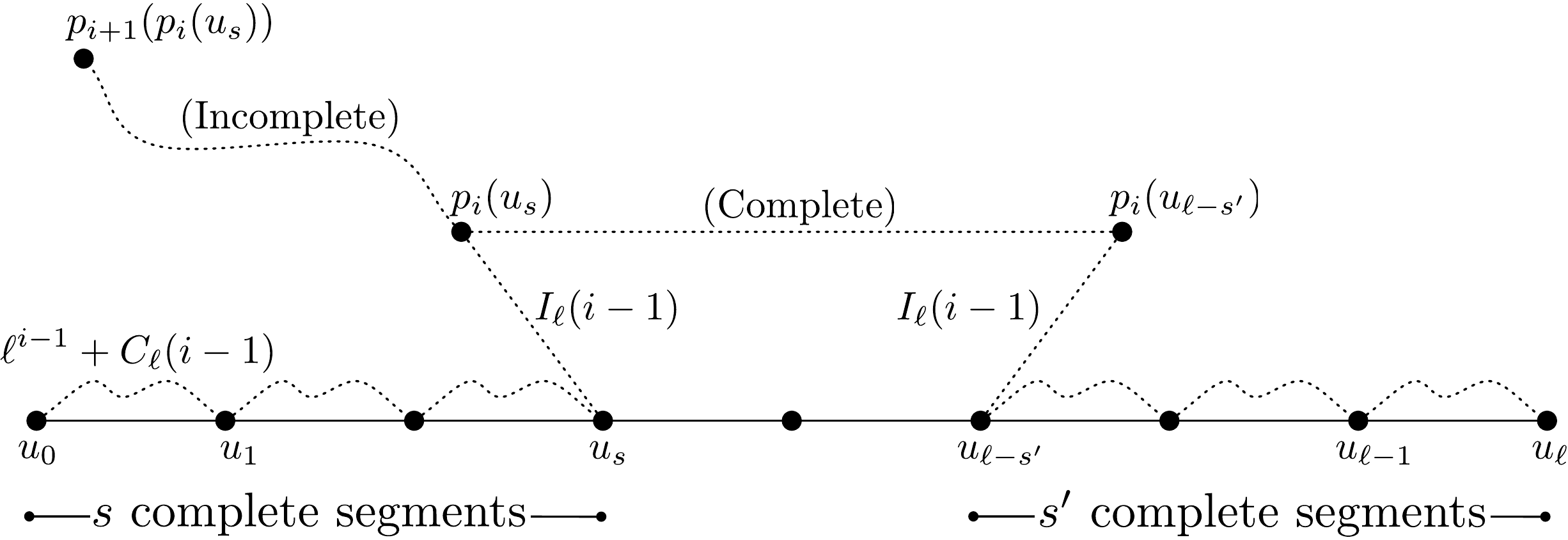}}
\caption{\label{fig:Complete-Incomplete}\small The shortest path from $u=u_0$ 
to $v=u_{\ell}$ has length $\ell^i$; it is partitioned into {\em segments} of length $\ell^{i-1}$.
A segment $(u_j,u_{j+1})$ is {\em complete} if $S$ contains a path of length $\ell^{i-1} + \Succ{\ell}{i-1}$
from $u_j$ to $u_{j+1}$ and {\em incomplete} if the distance from $u_j$ to $p_i(u_j)$ is at most $\Fail{\ell}{i-1}$.
If all segments are complete (not depicted) then $\dist_{S}(u_0,u_{\ell}) \le \dist(u_0,u_{\ell}) + \ell \cdot \Succ{\ell}{i-1}$.
If only the first $s$ segments and last $s'$ segments are complete and $p_i(u_{\ell-s'})$ lies in the ball 
$\Ball_{i+1}(p_i(u_s))$ then $S$ contains a path from $u_0$ to $u_{\ell}$ with length 
$\dist(u_0,u_{\ell}) + (s+s')\Succ{\ell}{i-1} + 4\Fail{\ell}{i-1}$.
On the other hand, if $p_i(u_{\ell-s'})$ lies outside $\Ball_{i+1}(p_i(u_s))$ then this gives
a bound on the distance from $p_i(u_s)$ to $p_{i+1}(p_i(u_s))$, and therefore a bound on $\dist(u_0,p_{i+1}(u_0))$.
From these cases we derive recursive expressions for $\Succ{\ell}{i}$ and $\Fail{\ell}{i}$.}
\end{figure}

When $i>0$, partition the shortest path from $u$ to $v$ into at most $\ell$ segments with length $\ell^{i-1}$, 
and let $u_j$ be the vertex on the path at distance $j\ell^{i-1}$ from $u$.  For the sake of simplicity, assume $\dist(u,v)=\ell^i$, so $v= u_{\ell}$.
Each segment from $u_j$ to $u_{j+1}$
is classified as either {\em complete} or {\em incomplete}.  If all segments are complete then 
$\dist_{S}(u,v) \le \dist_{S}(u,v) \le \dist(u,v) + \ell\cdot \Succ{\ell}{i-1}$.
If there is at least one incomplete segment, let there be $s$ complete segments on a prefix of the path 
and $s'$ complete segments on a suffix of the path, where $s+s' \le \ell-1$.  It follows that
\begin{align*}
\dist_{S}(u,p_{i}(u_s)) &\le \dist_{S}(u,u_s) + \dist_{S}(u_s,p_i(u_s))\\
					&\le \dist(u,u_s) + s\cdot \Succ{\ell}{i-1} + \Fail{\ell}{i-1}\\
\dist_{S}(v,p_{i}(u_{\ell-s'})) &\le \dist_{S}(v,u_{\ell-s'}) + \dist_{S}(u_{\ell-s'},p_i(u_{\ell-s'})) \\
						&\le \dist(v,u_{\ell-s'}) + s'\cdot \Succ{\ell}{i-1} + \Fail{\ell}{i-1}.
\intertext{If $p_i(u_{\ell-s'}) \not\in \Ball_{i+1}(p_i(u_s))$ then}
\dist_{S}(u,p_{i+1}(u)) &\le \dist(u,p_i(u_{s})) + \dist(p_i(u_s), p_i(u_{\ell-s'})) \\
				&\le (\ell-s')\ell^{i-1} + 3\Fail{\ell}{i}\\
				&\le \ell^i + 3\Fail{\ell}{i}		& \mbox{worst case when $s'=0$.}
\end{align*}
and the path from $u$ to $v$ is {\em incomplete}.  On the other hand, if $p_i(u_{\ell-s'}) \in \Ball_{i+1}(p_i(u_s))$
then $S$ contains a shortest (or nearly shortest, if $i=1$) path from $p_i(u_s)$ to $p_i(u_{\ell-s'})$, so 
\begin{align*}
\dist_{S}(u,v) &\le \dist_{S}(u,p_i(u_s)) + \dist_{S}(p_i(u_s), p_i(u_{\ell-s'})) \zero{+ \dist_{S}(p_i(u_{\ell-s'}), v)} \\
		&\le [s(\ell^{i-1} + \Succ{\ell}{i-1}) + \Fail{\ell}{i-1}]		& \mbox{from $u$ to $p_i(u_s)$}\\
		&\hcm[.5] + [(\ell-s-s')\ell^{i-1} + 2\Fail{\ell}{i-1} \;\; \{+2\}]	& \mbox{from $p_i(u_s)$ to $p_i(u_{\ell-s'})$}\\
		&\hcm[.5] + [s'(\ell^{i-1} + \Succ{\ell}{i-1}) + \Fail{\ell}{i-1}]	& \mbox{from $p_i(u_{\ell-s'})$ to $v$}\\
		&\le \dist(u,v) + (\ell-1)\Succ{\ell}{i-1} + 4\Fail{\ell}{i-1} \;\; \{+2\} &\mbox{worst case when $s+s'=\ell-1$}
\end{align*}
where the $\{+2\}$ is only present if $i=1$.
We satisfy Definition~\ref{def:succfail} by 
setting $\Succ{\ell}{1}=6, \Fail{\ell}{1}=\ell+3$, and, for $i\ge 2$,
$\Fail{\ell}{i} = \ell^i + 3\Fail{\ell}{i-1}$ and 
$\Succ{\ell}{i}$ to be the maximum of $\ell\cdot \Succ{\ell}{i-1}$ and $(\ell-1)\Succ{\ell}{i-1}+4\Fail{\ell}{i-1}$.
\end{proof}

We now find closed form bounds for $\Succ{\ell}{i}$ and $\Fail{\ell}{i}$.

\begin{lemma}\label{lem:succfail-closedform}
The values defined inductively in Lemma~\ref{lem:succfail-recursive} satisfy the following (in)equalities.
\begin{align*}
\Fail{2}{i} &= 3^{i+1} -  2^{i+1}\\
\Succ{2}{i} &\le 3^{i+1}\\
\Fail{3}{i} &= (i+1)3^i\\
\Succ{3}{i} &\le 4i3^i\\
\intertext{Define $c_{\ell} = \ell/(\ell-3)$.  For all $\ell\ge 4$ and $i\ge 1$,}
\Fail{\ell}{i} &\le c_{\ell} \ell^i		\\		
\Succ{\ell}{i} &\le \min\left\{\begin{array}{l}
4c_{\ell} \ell^i\\
(4c_{\ell}i+2)\ell^{i-1}
\end{array}\right.						
\end{align*}
\end{lemma}

\begin{proof}
All bounds are established by induction on $i$.  The cases when $\ell \in \{2,3\}$ are left as an exercise.
When $\ell \ge 4$ the base cases $i\in\{0,1\}$ clearly hold.  For incomplete paths and $i\ge 2$ we have
\begin{align*}
\Fail{\ell}{i} &= \ell^i + 3\cdot \Fail{\ell}{i-1} 				& \mbox{(by definition)}\\
		&\le \ell^i(1 + 3c_{\ell}/\ell)		\;\le\; c_\ell \ell^i	& \mbox{(induction hypothesis, $c_{\ell} = \f{\ell}{\ell-3}$)}
\intertext{and for complete paths we have two cases,}
\Succ{\ell}{i} &= (\ell-1)\Succ{\ell}{i-1} + 4\Fail{\ell}{i-1}	& \mbox{(by definition)}\\
		&\le (\ell-1)4c_{\ell}\ell^{i-1} + 4c_{\ell}\ell^{i-1} 	& \mbox{(1st induction hypothesis)}\\
		&= 4c_{\ell}\ell^i\\
\mbox{and }\; &\le (\ell-1)(4c_{\ell}(i-1)+2)\ell^{i-2} + 4c_{\ell}\ell^{i-1} 	& \mbox{(2nd induction hypothesis)}\\
		&\le (4c_{\ell} i+2)\ell^{i-1}.
\end{align*}
\end{proof}

Observe that when we check whether $p_i(u_{\ell-s'}) \in \Ball_{i+1}(p_i(u_s))$, $i\ge 2$,
the distance between $p_i(u_{\ell-s'})$ and $p_i(u_{s})$ is maximized when $s=s'=0$;
it is at most 
\[
\ell^i + 2\Fail{\ell}{i-1} = \ell^i + 2c_{\ell}\ell^{i-1}  \; < (\ell+2)^i.
\]
Thus, as long as $\ell \le r$, the criterion $p_i(u_{\ell-s'}) \in \Ball(p_i(u_s), (r+2)^i)$ will also hold.
This retroactively justifies the constraint $\ell \le r$ in Lemma~\ref{lem:succfail-recursive}.

\begin{theorem}
The spanner $S(k,r)$ has size $O(r^hkn^{1+\f{1}{2^{k+1}-1}})$, where $h = \frac{3\cdot 2^{k-1} - (k+2)}{2^{k+1}-1} < 3/4$.
Its stretch changes as a function of the distance $d$ being approximated.
\begin{itemize}
\item For $d\ge 2^{k}$ it is a multiplicative $O((3/2)^k)$-spanner.
\item For $d\ge 3^{k}$ it is a multiplicative $O(k)$-spanner.
\item For $d\ge \ell^{k}$, $\ell\in[4,k)$, it is a multiplicative 
	$(5+O(1/\ell))$-spanner, and when $\ell \in [k,r]$ 
	it is a multiplicative $(1 + (4k+O(1))/\ell)$-spanner.
\end{itemize}
$S(k,r)$ is a $(1+\epsilon, ((4k+O(1))/\epsilon)^{k-1})$-spanner for every $\epsilon$ such that 
$(4k+O(1))/\epsilon < r$.  Its stretch function can also be expressed as
$f(d) = d + (4+o(1))kd^{1-\f{1}{k}} + 3^k$ for {\em all} $d \le r^k$,
and $f(d) = d + (4+o(1))kd/r$ for larger $d$, where the $o(1)$s go to zero as $d$ increases.
\end{theorem}

\begin{proof}
Let $\dist(u,v) \ge 2^{k}$ be the distance to be approximated and $\ell = \floor{d^{\f{1}{k}}}$, where $\ell \le r$.
Partition the shortest path $P(u,v)$ into intervals of length precisely $\ell^{k-1}$, with at most one shorter interval.
Since $d < (\ell+1)^{k}$, there are between $\ell$ and $\floor{(\ell+1)(1+1/\ell)^{k-1}}$ intervals.
If all intervals are complete then $\dist_S(u,v) \le d + \ceil{\frac{d}{\ell^{k-1}}}\Succ{\ell}{k-1}$.  
If at least one is incomplete then $\dist_S(u,v) \le d + (\ceil{\frac{d}{\ell^{k-1}}}-1)\Succ{\ell}{k-1} + 4\Fail{\ell}{k-1}$.
If $\ell \in [3,k-1]$ then according to Lemma~\ref{lem:succfail-closedform}, 
$\Succ{\ell}{k-1}=4\Fail{\ell}{k-1}$ and we are indifferent between these two possibilities.
If $\ell \ge k$ or $\ell=2$ then $\Succ{\ell}{k-1} < 4\Fail{\ell}{k-1}$, so the second case is worse.
When $\ell=2$ we have
\begin{align*}
\lefteqn{d + \paren{\ceil{\frac{d}{2^{k-1}}}-1}\Succ{2}{k-1} + 4\Fail{2}{k-1}}\\
&<  d + \paren{\ceil{\frac{d}{2^{k-1}}}-1}3^{k} + 4\cdot 3^{k}\\
&\le d(1 + 3(3/2)^{k-1}) + 4\cdot 3^{k}\\
\intertext{So $S(k,r)$ is a multiplicative $O((3/2)^{k-1})$-spanner for $d\ge 2^k$.  This is a non-trivial multiplicative stretch.  Traditional multiplicative stretch spanners with size $n^{1+\f{1}{2^{k+1}-1}}$ size stretch some pairs by a factor of $2^{k+2}-3$.
When $\ell=3$ we have}
\lefteqn{d + \paren{\ceil{\frac{d}{3^{k-1}}}-1}\Succ{3}{k-1} + 4\Fail{3}{k-1}}\\
&<	d + \paren{\ceil{\frac{d}{3^{k-1}}}-1}4(k-1)3^{k-1}	+ 4k3^{k-1}\\
&\le d(1 + 4(k-1)/3 + 4k/3)				& \mbox{(since $3^{k-1} \le d/3$)}
\intertext{Thus $S(k,r)$ functions as a multiplicative $O(k)$-spanner when $d\ge 3^{k}$.  When $\ell \in [4,k)$,}
\lefteqn{d + \paren{\ceil{\frac{d}{\ell^{k-1}}}-1}\Succ{\ell}{k-1} + 4\Fail{\ell}{k-1}}\\
&< 	d + \paren{\ceil{\frac{d}{\ell^{k-1}}}-1}4c_{\ell} \ell^{k-1} + 4c_{\ell}\ell^{k-1}\\
&< 	d(1 + 4c_{\ell}  + 4c_{\ell}/\ell) \\
&= d\paren{1+\f{4(\ell+1)}{\ell-3}} = (5+O(\fr{1}{\ell}))d & \mbox{(since $\ell^{k-1} \le d/\ell$, $c_\ell = \ell/(\ell-3)$)}
\intertext{The multiplicative stretch of $S(k,r)$ tends to $5$ as $d$ increases from $3^{k}$ to $(k-1)^{k}$.
When $\ell \ge k$ we have}
\lefteqn{d + \paren{\ceil{\frac{d}{\ell^{k-1}}}-1}\Succ{\ell}{k-1} + 4\Fail{\ell}{k-1}}\\
&\le d + \paren{\ceil{\frac{d}{\ell^{k-1}}}-1}(4c_{\ell}(k-1)+2)\ell^{k-2} \zero{\:+\: 4c_{\ell}\ell^{k-1}}\\
&\le d\paren{1+\frac{4c_\ell k + 2}{\ell}}		& \mbox{$\ell^{k-1} \le d/\ell$}
\end{align*}
When $d \ge \ell^k \ge k^k$
the multiplicative stretch is $1 + (1+o(1))4k/\ell$, 
where the $o(1) = O(1/\ell)$ tends to zero as $\ell$ increases.
When $\ell \ge (4c_{\ell}k + 2)/\epsilon$ the multiplicative stretch becomes $1+\epsilon$.

One may confirm that by setting $\ell = \floor{d^{\f{1}{k}}}$, in all the cases above
the stretch function of $S(k,r)$ can be expressed as 
$f(d) = d + O(kd^{1-\f{1}{k}} + 3^k)$, for $\ell^k \le d \le r^k$,
and $f(d) = d + O(kd/r)$ for $d \ge r^k$.
The leading constants in the terms $O(kd^{1-\f{1}{k}})$ and $O(kd/r)$ tend to $4$ as $d$ increases.
\end{proof}

Setting $r = (4k + O(1))/\epsilon$, we obtain a $(1+\epsilon, O(k/\epsilon)^{k-1})$-spanner with size 
$O((k/\epsilon)^h kn^{1+\f{1}{2^{k+1}-1}})$.
This spanner is sparsest when $\epsilon>0$ is a fixed constant and $k = \log_2\log_2 n - O(1)$: 
it is then a $(1+\epsilon, ((4+o(1))\log\log n)^{\log\log n-O(1)})$-spanner with size $O(n(\log\log n)^{7/4})$.
When $k=\log\log n$ it is possible to reduce the size of this spanner to 
$O(kn + nr^{3/4}) = O(n(\log\log n + (\epsilon^{-1}\log\log n)^{3/4}))$.
The $kn$ term reflects the cost of the paths $\{P(u,p_i(u))\}_{u\in V, i\in [1,k]}$.  
Rather than equalize the remaining contribution of $E_0',\ldots,E_{k}'$, one chooses 
the sampling probabilities such that $|\tilde{E}_1|$ and $|E_0'|$ are balanced 
and $|E_2'|,|E_3'|,\ldots,|E_{k}'|$ decay geometrically.

Even sparser $(1+\epsilon,\beta)$-spanners are known, but they have slightly worse tradeoffs.  Pettie~\cite{Pettie-Span09}
constructed a $(1+\epsilon, O(\epsilon^{-1}\log\log n)^{\log \log n})$-spanner
with size $O(n\log\log(\epsilon^{-1}\log\log n))$.

\section{Lower Bounds for Hopsets}\label{sect:hopset}

In this section, we show lower bounds on the tradeoffs between $\beta$ and $\eps$ in $(\beta,\epsilon)$-hopsets, 
subject to an upper bound on the number of edges in the hopset.
We begin by making some minor modifications to the construction of the lower bound graphs $\{\Graph_k\}_k$
from Section~\ref{sect:lower-bound}, then prove lower bounds on hopsets for $\Graph_k$.

\subsection{A New Construction of $\Graph_k$}

In the base case $k=1$, redefine $\Graph_1[p]$ to be a copy of $\Base[p]$ on $\ell+1$ layers (rather than a biclique), 
each edge of which has unit length.  Naturally $\Pairs_1$ is $\Pairs(\Base[p])$.  
Rather than have $p^2$ edges and $p^2$ pairs in its pair-set, 
the new $\Graph_1[p]$ has $p^{2-o(1)}$ edges and $p^{2-o(1)}$ pairs in $\Pairs_1$, when $\ell = p^{o(1)}$.
The graph $\Graph_k[p]$ is formed as before, by taking a copy of 
$\DblBase[p]$ and replacing each vertex in an interior layer with a standard or reversed 
copy of $\Graph_{k-1}[p']$, where $p' = p/\Loss_\ell(\sqrt{p})$.  
Rather than subdivide edges of $\DblBase[p]$
into paths of length $(2\ell-1)^{k-1}$, we leave them as is, but give them weight $(2\ell-1)^{k-1}$.
The construction of $\Pairs_k$ from $\Pairs_{k-1}$ is exactly as in Section~\ref{sect:lower-bound}.
When $\ell = p^{o(1)}$, the size of $\Graph_k[p]$ and $\Pairs_k[p]$ only differ from the 
old $\Graph_k[p]$ and $\Pairs_k[p]$ (from Section~\ref{sect:lower-bound}) by $p^{o(1)}$ factors.

\begin{lemma} \label{lem:hopset-unique-paths}
If $(u,v) \in \Pairs_k$ then there is a \emph{unique} shortest path from $u$ to $v$ in $\Graph_k$.
If $k=1$ the path has length exactly $\ell$ and if $k\ge 2$ the path has length exactly $(2k-1)\ell (2\ell-1)^{k-1}$
and passes through $2\ell-1$ copies of $\Graph_{k-1}$.
\end{lemma}
\begin{proof}
The proof follows the same lines as Lemma \ref{lem: unique paths}.  Let $d_k$ be the distance between the 
input ports and output ports in $\Graph_k$.  Then
\begin{align*}
d_1 &= \ell\\
d_k &= (2\ell-1)d_{k-1} + 2\ell (2\ell-1)^{k-1}
\end{align*}
and the claim follows by induction on $k$.
\end{proof}

\subsection{Simplifying the Hopset $H$}

Consider a hopset $H$ for $\Graph_k$.  In order to simplify the arguments to come we will manipulate 
$H$ so that it satisfies certain structural properties. 

\begin{definition}
Let $H$ be a hopset of $\Graph_k$.
\begin{enumerate}
\item An edge $(u,v) \in H$ has \emph{order $i$}, $1\le i \le k$, if $u$ and $v$ are contained in a single copy of $\Graph_i$ within $\Graph_k$.
\item Suppose $(u,v)\in H$ has order $i$.  If $u$ and $v$ are in adjacent copies of $\Graph_{i-1}$ 
(or $u$ is in a copy of $\Graph_{i-1}$ and $v$ is an adjacent input/ouput port of the copy of $\Graph_i$ containing it) 
then $(u,v)$ is \emph{short}.  Otherwise $(u,v)$ is \emph{long}.
\end{enumerate}
\end{definition}

Later it will be convenient to assume that $H$ contains only long edges.  
Lemma~\ref{lem:simple-h} shows that short
edges can be expunged from $H$ without affecting $\beta$ and $|H|$ by more than a constant factor.

\begin{lemma} \label{lem:simple-h}
Let $H$ be a $(\beta, \eps)$ hopset for $\Graph_k$.
Then there is a $(O(k\beta), \eps)$ hopset $H'$ for $\Graph_k$
containing only long edges, with $|H'| \le 2|H|$.
\end{lemma}
\begin{proof}
Let $(u,v) \in H$ be an order $i$ short edge connecting adjacent copies of $\Graph_{i-1}$, and let $(\hat{u},\hat{v})$
be the edge joining these copies.  See Figure~\ref{fig:hopset}.  Replace $(u,v)$ in $H$ with edges $(u,\hat{u}), (\hat{v},v)$.
Any path formerly using $(u,v)$ can now use three edges in its place: $(u,\hat{u}), (\hat{u},\hat{v}), (\hat{v},v)$.
Observe that $(u,\hat{u}), (\hat{v},v)$ are order $i-1$ edges, which may be \emph{short} order $i-1$ edges.
If $(u,\hat{u})$ and/or $(\hat{v},v)$ are short, recursively process them in the same way.  
Whereas processing $(u,v)$ spawned two edges, processing $(u,\hat{u})$ or $(\hat{v},v)$ spawns a single
edge since $\hat{u}$ and $\hat{v}$ are input/output ports of copies of $\Graph_{i-1}$, and not contained in any copy of $\Graph_{i-2}$.
Thus, after this recursive process completes, each original edge $(u,v)$ is simulated by 
a path with at most $O(k)$ hops.
\begin{figure}
\centering
\scalebox{.4}{\includegraphics{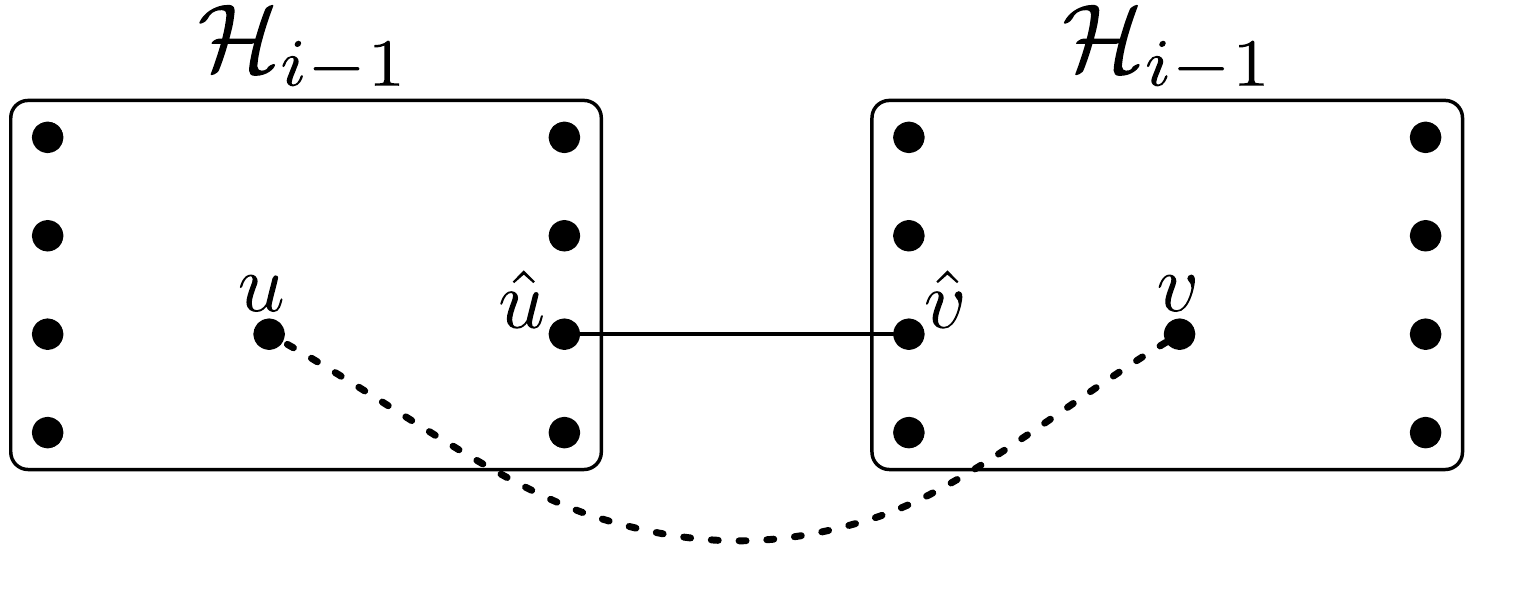}}\\
\scalebox{.4}{\includegraphics{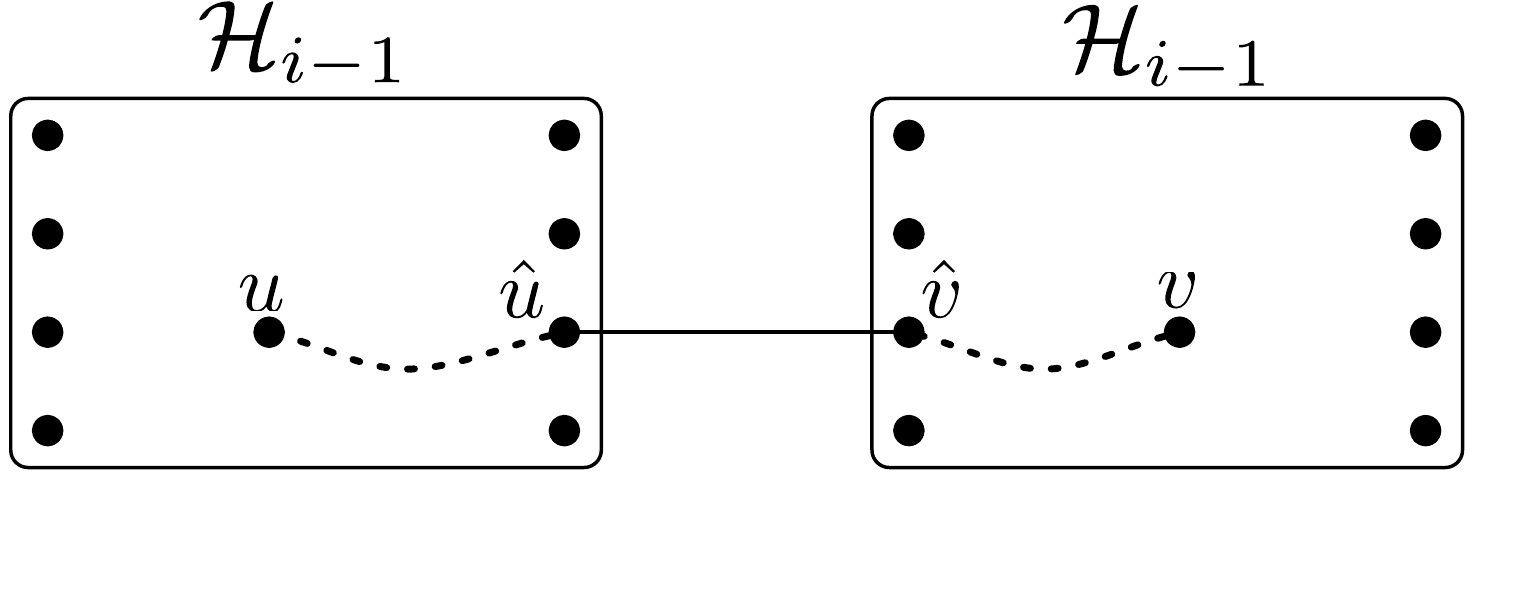}}
\caption{\label{fig:hopset}Above: an order $i$ short edge $(u,v)$ joining two vertices in {\em adjacent}
copies of $\Graph_{i-1}$, the edge joining these copies being $(\hat{u},\hat{v})$.  Below: replacing $(u,v)$ with three hops $(u,\hat{u}),(\hat{u},\hat{v}),(\hat{v},v)$.  If (order $i-1$) edges $(u,\hat{u}),(\hat{v},v)$ are still short, then are processed recursively.  Note: $\hat{u}$ and $\hat{v}$ are input/output ports in copies of $\Graph_{i-1}$; only $u$ and $v$ may be contained in copies of $\Graph_{i-2}$.}
\end{figure}
\end{proof}

Henceforth we only consider hopsets for $\Graph_k$ that contain only long edges.

\subsection{Tradeoffs Between $\beta$ and $\epsilon$}

We next assign \emph{ownership} of each long order $i$ edge $(x, y) \in H$ to a pair 
in $\Pairs_k$.  Suppose that $G_{i-1}^x$ and $G_{i-1}^y$ are the copies of $\Graph_{i-1}$ containing $x$ and $y$, respectively.
Let $P \in \Pairs_k$ \emph{own} $(x,y)$ if the unique shortest path for $P$ intersects both $G_{i-1}^x$ and $G_{i-1}^y$.  
It is not obvious how to assign ownership over short edges.
Lemma~\ref{lem:hopset-own} motivates our procedure for expunging short edges by showing
that each remaining long edge is owned by at most one pair in $\Pairs_k$.

\begin{lemma}\label{lem:hopset-own}
Each long edge $(x, y)\in H$ is owned by at most one pair in $\Pairs_k$.
\end{lemma}
\begin{proof}
Suppose $(x, y)$ has order $i$.  Let $G_i$ be the copy of $\Graph_i$ containing $x,y$ and 
$G_{i-1}^x,G_{i-1}^y$ be the copies of $\Graph_{i-1}$ within $G_i$
containing $x$ and $y$.
Each pair in $\Pairs_k$ has a unique shortest path in $\Graph_k$; if it intersects $G_i$ then it enters and exits $G_i$ by a unique (input port, output port) pair,
which is included in $\Pairs_{i}$.
Thus, it suffices to prove that at most one pair in $\Pairs_i$ has a shortest path intersecting both $G_{i-1}^x, G_{i-1}^y$.
Since $(x,y)$ is long, $G_{i-1}^x,G_{i-1}^y$ are not adjacent, i.e., the corresponding nodes $\bar{x}, \bar{y}$ in $\DblBase$ are at distance at least 2.
In order for a path in $\Pairs_i$ to intersect $G_{i-1}^x,G_{i-1}^y$ the edges on the path between $\bar{x}$ and $\bar{y}$ must be labeled alternately with
two labels $a,b$.  The triple $\bar{x},a,b$ uniquely determines the input port and output port in $G_i$, and therefore uniquely determines 
a member of $\Pairs_k$ that owns $(x,y)$.  If the shortest path between $\bar{x}$ and $\bar{y}$ is not labeled alternately with two labels $a,b$, then no pair in $\Pairs_k$ owns $(x,y)$.
\end{proof}

If the size of the hopset $H$ is strictly less than $|\Pairs_k|$ then some pair in $\Pairs_k$ must not own any edges.
Lemma~\ref{lem:non-ownership-error} shows that for any pair with this property, it is impossible to get below additive
error $2(\ell+1)^{k-1}$ via a path having at most $(\ell-1)^k$ hops.

\begin{lemma} [Compare to Lemma \ref{lem:removing-critical-edges}] \label{lem:non-ownership-error}
Let $H$ be a hopset for $\Graph_k$ containing only long edges and let $(u, v) \in \Pairs_k$ be a pair
that owns no edges in $H$.
Then we have
\[
\dist^{((\ell-1)^k)}_{\Graph_k \cup H}(u, v) \ge \dist_{\Graph_k}(u, v) + 2(\ell+1)^{k-1}.
\]
\end{lemma}
\begin{proof}
The proof is by induction over $k$.

\paragraph{Base Case.}
When $k=1$, $\Graph_1$ is a layered bipartite graph, so we either have
\[
\dist^{(\ell - 1)}_{\Graph_k \cup H}(u,v) = \dist_{\Graph_k}(u,v)
\]
or
\[
\dist^{(\ell - 1)}_{\Graph_k \cup H}(u,v) \ge \dist_{\Graph_k}(u,v) + 2
\]
so it suffices to rule out the former possibility.
We have $\dist_{\Graph_1}(u,v) = \ell$; thus, the shortest $u$--$v$ path in $\Graph_k \cup H$ using at most $\ell-1$ hops must include at least one edge in $H$.
All edges in $H$ have order $1$.  Since $(u,v)$ owns no edge in $H$, there is no edge $(x, y) \in H$ with $x, y$ on the unique shortest $u$--$v$ path.
It follows that $\dist^{(\ell-1)}_{\Graph_1 \cup H}(u,v) \ne \dist_{\Graph_1}(u,v)$ and the base case is complete.

\paragraph{Inductive Step.}
We now argue the inductive step.
Let $U$ be a $u$--$v$ path in $\Graph_k \cup H$ that uses at most $(\ell - 1)^k$ hops
and $P$ be the corresponding path in $\Graph_k$, i.e., the one obtained by replacing each $H$-edge in $U$ with a shortest path between its endpoints.
Finally, define $\hat{P}$ to be the projection on $P$ onto $\DblBase[p]$.
We consider two cases depending on whether $U$ uses at least one order $k$ edge from $H$ or not.

\paragraph{Inductive Step -- Case 1:} Suppose that $U$ includes an edge $(x, y) \in H$ of order $k$.
In this case we do not need the inductive hypothesis.
Since the pair $(u, v)$ does not own $(x, y)$, this means that $x$ (or $y$) is in a copy of 
$\Graph_{k-1}$ that is disjoint from the unique shortest $u$--$v$ path in $\Graph_k$.
Thus, $\hat{P}$ is not equal to the unique shortest $u$--$v$ path in $\DblBase$.
Since $\DblBase$ is bipartite, the length of $\hat{P}$ is at least $2 + \dist_{\DblBase}(u,v)$.
Each of these two edges has weight $(2\ell-1)^{k-1}$ in $\Graph_k$, so
$P$ (and $U$) have length at least 
$\dist_{\Graph_k}(u,v) + 2(2\ell - 1)^{k-1}$.
This same analysis applies whenever $\hat{P}$ is not identical to the shortest $u$--$v$ path in $\DblBase$.

\paragraph{Inductive Step -- Case 2:} Suppose that $U$ contains no edges of $H$ with order $k$.
By the above analysis, we can restrict our attention to the case when $\hat{P} = (u=u_0,u_1,\ldots,u_{2\ell} = v)$ 
is the shortest $u$--$v$ path in $\DblBase$.  Let $G(u_j)$ be the copy of $\Graph_{k-1}$ substituted for $u_j$ and $U_j$.
It follows that $U$ contains all the weighted edges joining consecutive $G(u_{j-1}),G(u_{j})$, and some paths $\{U_j\}$ 
joining an input port and output port of $G(u_j)$.  Moreover, these input/output port pairs must be in $\Pairs_{k-1}$.

Partition the $\{U_j\}_{1\le j\le 2\ell-1}$ based on whether their hop count is at most $(\ell-1)^{k-1}$
or at least $(\ell-1)^{k-1}+1$.  There can be at most $\floor{\frac{(\ell-1)^{k} - 2\ell}{(\ell-1)^{k-1}+1}} = \ell-2$ 
subpaths in the second category, meaning at least $(2\ell-1) - (\ell-2) = \ell+1$ of the subpaths in $\{U_j\}$ use at most $(\ell-1)^{k-1}$ hops.
Applying the inductive hypothesis to these subpaths, we have 
\[
\dist_{\Graph_k\cup H}^{((\ell-1)^{k})}(u,v) \ge \dist_{\Graph_k}(u,v) + (\ell+1)\cdot 2(\ell+1)^{k-2} = \dist_{\Graph_k}(u,v) + 2(\ell+1)^{k-1}.
\]
\end{proof}

We are finally ready to show:
\begin{theorem}\label{thm:hopset}
Fix a positive integer $k$ and parameter $\epsilon > 1/n^{o(1)}$.
Any construction of $(\beta, \eps)$-hopsets with size 
$n^{1 + \frac{1}{2^{k} - 1} - \delta}$, $\delta > 0$, 
has $\beta = \Omega_k\left( \frac{1}{\eps} \right)^{k}$.
\end{theorem}
\begin{proof}
Construct $\Graph_k[p]$ with respect to an $\ell = p^{o(1)}$ to be determined, so $|\Pairs_k| = n^{1+\frac{1}{2^k-1} - o(1)}$.
Let $H$ be a $(\beta, \eps)$ hopset for $\Graph_k$ containing only long edges.
If $|H| \le n^{1 + \frac{1}{2^{k} - 1} - \delta}$ for some $\delta > 0$ then $|H| < |\Pairs_k|$, meaning 
some pair $(u,v)\in\Pairs_k$ owns no $H$ edges.
By Lemma \ref{lem:non-ownership-error}, we then have
\[
\dist^{(\ell-1)^k}_{\Graph_k \cup H}(u,v) = \dist_{\Graph_k}(u,v) + 2(\ell+1)^{k-1}.
\]
By Lemma \ref{lem:hopset-unique-paths}, 
$\dist_{\Graph_k}(u,v) = (2k-1)\ell (2\ell-1)^{k-1}$.
Thus the relative error $\hat{\epsilon}$ of any $(\ell-1)^{k}$-hop path is 
\[
\hat{\epsilon} = \frac{2(\ell+1)^{k-1}}{(2k-1)\ell(2\ell-1)^{k-1}} 
> \frac{2}{(2k-1)\ell 2^{k-1}}.
\]
We choose $\ell$ as a function of $k$ and $\epsilon$ so that 
$\ell < \frac{1}{2^{k-2}(2k-1)\epsilon}$, which implies $\hat{\epsilon} > \epsilon$.  
In order for $H$ to be a $(\beta,\epsilon)$-hopset for $\Graph_k$, it must be that $\beta > (\ell-1)^k = \Omega_k(1/\epsilon)^{k}$.
\end{proof}

Observe that Theorem~\ref{thm:hopset} implies several interesting corollaries: 
any $(\beta,\eps)$-hopset with $\beta=o(1/\epsilon)$ must have size $\Omega(n^{2-o(1)})$
and any such hopset with $\beta = o(1/\epsilon^2)$ must have size $\Omega(n^{4/3-o(1)})$.

\begin{remark}
The analysis of Lemma~\ref{lem:non-ownership-error} still has some slack in it, 
which introduces the $2^{k-1}$ factor error in Theorem~\ref{thm:hopset}.
A more careful analysis will most likely reduce the hop lower bound to $\beta = \Omega(1/(k\epsilon))^k$,
mimicking the dependency on $k$ from Theorems~\ref{thm:lb-onepluseps} and \ref{thm:lb-DO}.
\end{remark}

\begin{remark}
The construction of $\Graph_k$ from this section was essentially the same as that of Section~\ref{sect:lower-bound}, except in the base case.  Had we used Section \ref{sect:lower-bound}'s definition of $\Graph_1[p]$ (a $K_{p,p}$ biclique, rather than a copy of $\Base[p]$), Theorem~\ref{thm:hopset} would have arrived at a weaker lower bound on $\beta = \Omega_k(1/\epsilon)^{k-1}$.
\end{remark}

\section{Lower Bounds on Compressing High Girth Graphs\label{sect:girth-lbs}}

The density of the graph $\Graph_k$ constructed in Section~\ref{sect:lower-bound} 
comes exclusively from complete bipartite graphs (copies of $\Graph_1$), that is, $\Graph_k$ has girth 4.
This feature of the construction turns out to be absolutely essential.
Baswana et al.~\cite{BKMP10} showed that the class of graphs with girth (length of the shortest cycle)
larger than 4 contains additive spanners below the $4/3$ threshold.
For example, graphs with girth 5 contain additive $12$-spanners with size $O(n^{6/5})$.

\begin{theorem} [\cite{BKMP10}] \label{thm: bkmp girth}
For any integer $\gamma \ge 1$, any graph with girth at 
least $2\gamma+1$ contains an additive $6\gamma$-spanner on 
$O(n^{1 + \f{1}{2\gamma+1}})$ edges.
\end{theorem}

In this section, we extend our lower bound technique to show that the exponent of Theorem~\ref{thm: bkmp girth}
is optimal.  More generally, we establish a hierarchy of tradeoffs for sublinear additive graph compression schemes 
that depend on $k$ and $\gamma$.  From a technical point of view, this section highlights two degrees of freedom
that were not used in Section~\ref{sect:lower-bound} or~\cite{AbboudB16}.  We use alternative base-case 
graphs (rather than bicliques $\Graph_1$) and form $\DblBase[p]$ from an {\em imbalanced} product 
of $\Base[p_1]$ and $\Base[p_2]$, where $p=p_1p_2$.
Our construction uses a slightly stronger, but equivalent, statement of the Girth Conjecture 
that asserts a lower bound on the degree rather than the total size.  

\begin{conjecture} [Girth Conjecture~\cite{Erdos63,BondyS74,Bollobas78}]
For any integer $\gamma \ge 1$, there exists a graph with $n$ vertices, girth $2\gamma + 2$, and 
minimum degree $\Omega(n^{1/\gamma})$.
\end{conjecture}

Our goal is to prove the following theorem.  Observe that by setting $k=2$, 
Theorem~\ref{thm: girth lower bounds} implies that the exponent of Theorem \ref{thm: bkmp girth} cannot be improved.

\begin{theorem} \label{thm: girth lower bounds}
Fix integers $\gamma \ge 1, k \ge 2$. Consider any data structure that answers approximate distance queries for the class of $n$-vertex undirected graphs with girth at least $2\gamma + 1$. 
Assuming the Girth Conjecture, if the stretch function of the data structure is
$$f(d) < d + c_{k,\gamma} d^{1 - 1/k}, \; \mbox{ for $c_{k,\gamma} \approx \fr{2}{(\gamma(k-1))^{1-1/k}}$ and $d$ sufficiently large}$$ 
then on some graph the data structure occupies at least $\Omega(n^{1 + \frac{1}{(\gamma + 1)2^{k-1} - 1} - o(1)})$ bits.
\end{theorem}

We remark that this theorem holds trivially for any super-constant $\gamma$ (with a sufficiently large $n^{-o(1)}$ factor), since the lower bound becomes $\Omega(n)$.
Thus, we treat $\gamma$ as a fixed constant throughout this section.
The remainder of this section constitutes a proof of Theorem~\ref{thm: girth lower bounds}.
We first make a simple observation about the hypothesized graphs from the Girth Conjecture.

\begin{observation}\label{lem: large girth sets}
Fix a $\gamma \ge 1$.  If $G$ has girth $2\gamma+2$ and minimum degree $\Omega(n^{1/\gamma})$, 
there are $\Omega(n^{d/\gamma})$ nodes at distance exactly $d$ from any node $u$, for any $0 \le d \le \gamma$.
\end{observation}

Recall that the reason $\Graph_1[p]$ from Section~\ref{sect:lower-bound} was useful was because its edge-set was 
{\em covered} by {\em unique}, {\em disjoint} shortest paths between $p$ input ports and $p$ output ports.
We will prove something analogous for high-girth graphs, but with these properties weakened slightly, in ways that have no
adverse effect on the overall construction.

\begin{definition} [\cite{AbboudB16-SODA}]
In a graph $G$, we say that a pair of nodes $s, t$ \emph{requires} an edge $e$ if every shortest path from $s$ to $t$ includes $e$.
\end{definition}

\begin{lemma} \label{lem: hard girth graphs}
Fix integers $\gamma\ge 1$ and $0 \le i < \gamma$. Assuming the Girth Conjecture, 
there is a graph $H = (V, E)$ on $n$ vertices and $\Omega(n^{1 + 1/\gamma})$ edges 
with girth $2\gamma+2$, and disjoint node subsets $S, T \subseteq V$ of sizes
$|S| = \Omega(n^{i/\gamma}), |T| = \Omega(n^{(\gamma + 1 - i)/\gamma})$
such that each edge $e \in E$ is required by some $(s,t)\in S\times T$ with $\dist(s,t)=\gamma$.
\end{lemma}
\begin{proof}
Let $H$ initially be any graph with minimum degree $\Omega(n^{1/\gamma})$ and girth $2\gamma+2$.
Sample node subsets $S, T$ independently and uniformly at random of the appropriate size
and let $P\subset S\times T$ be such that $(s,t)\in P$ if and only if $\dist(s,t) = \gamma$.
The shortest $s$--$t$ path is unique, due to $H$'s girth, 
so all its edges are required by $(s,t)$.  Discard from $H$ all edges not required by any pair in $P$.

We will now prove that any particular edge remains in $H$ with constant probability, so there exists some choice of $S,T$ for which
at least a constant fraction of the edges are retained.
Let us consider an arbitrary edge $(u, v)$ in $H$.
If there exist nodes $s \in S, t \in T$ such that $\dist_{H-\{(u,v)\}}(s, u) + \dist_{H-\{(u,v)\}}(t, v) = \gamma - 1$, 
then $\dist(s,t) = \gamma$ and the unique $s$--$t$ shortest path requires $(u,v)$.

Let $A$ be all vertices at distance exactly $\gamma-i$ from $u$ in $H-\{(u,v)\}$ and $B$ be all vertices at distance exactly $i-1$ from $v$
in $H-\{(u,v)\}$.  By Observation \ref{lem: large girth sets}, 
$|A| = \Omega(n^{\frac{\gamma - i}{\gamma}})$ 
and $|B| = \Omega(n^{\frac{i-1}{\gamma}})$.  Since $|A||S|$ and $|B||T|$ are both $\Omega(n)$,
with constant probability $A\cap S \neq \emptyset$ and $B\cap T \neq \emptyset$.
\end{proof}

\begin{lemma} \label{lem: girth pairs}
Let $H = (V, E)$ be a graph and $S, T \subseteq V$ be node subsets as described in Lemma \ref{lem: hard girth graphs}.
Then there exists a set $\Pairs_1^{\gamma} \subseteq S \times T$ of size $|\Pairs_1^{\gamma}| = \Omega(n^{1 + 1/\gamma} / \gamma^3)$, as well as a mapping $\phi: \Pairs_1^{\gamma} \to E$, with the following two properties:
\begin{itemize}
\item For each $(s,t) \in \Pairs_1^{\gamma}$, the pair $(s,t)$ requires the edge $\phi(s,t)$.
\item For each $(s, t) \in \Pairs_1^{\gamma}$, the unique shortest path from $s$ to $t$ in $H$ contains no edge $e$ such that $e=\phi(s',t')$
for some $(s',t')\neq (s,t)$.
\end{itemize}
\end{lemma}
\begin{proof}
Let $P \subseteq S \times T$ be the set of $s, t$ pairs for which $\dist(s, t) = \gamma$.
Since $P = \Theta(n^{1 + 1/\gamma})$, each node pair in $P$ has distance $\gamma$, 
and $H$ has $\Omega(n^{1 + 1/\gamma})$ edges, it follows that the average edge in $H$ is required by 
$c = O(\gamma)$ different pairs in $P$.  By Markov's inequality, at most half the edges in $H$ are required by more than 
$2c$ pairs; let $F$ be the set of edges required by $2c$ or fewer node pairs in $S \times T$.
We then have $|F| = \Omega(n^{1 + 1/\gamma})$.

We now build $\Pairs_1^{\gamma}$ and $\phi$ using the following process.
Iterate through the pairs in $P$ in any order.
For each $(s, t) \in P$, let $Q$ be the unique shortest path from $s$ to $t$.  
If $F\cap Q = \emptyset$, i.e., if $(s,t)$ requires no $F$-edges, then discard $(s,t)$ from $P$.
Otherwise, include $(s,t)$ in $\Pairs_1^\gamma$, set $\phi(s, t)$ to be any edge in $F\cap Q$,
and discard from $P$ any other pair that requires any edge in $F\cap Q$.  Since $|Q| \le \gamma$ 
and edges in $F$ are required by at most $2c$ pairs, we discard $O(c\gamma)$ pairs for each $(s,t)$ not discarded.

The necessary properties of $\Pairs_1^{\gamma}$ are immediate from the construction.
To bound the size of $|\Pairs_1^{\gamma}|$, first note that 
$\Omega(n^{1 + 1/\gamma} / \gamma)$ pairs in $P$ require at least one edge in $F$, since $|F| = \Omega(n^{1 + 1/\gamma})$. 
Of these $\Omega(n^{1 + 1/\gamma} / \gamma)$ node pairs, each one added to $\Pairs_1^\gamma$ causes
at most $O(c\gamma)=O(\gamma^2)$ to be discarded, so $|\Pairs_1^\gamma| = \Omega(n^{1 + 1/\gamma} / \gamma^3)$.
\end{proof}

\paragraph{The Lower Bound Construction.}

The graph $\Base[p]$ is defined exactly as before, and the parameter $\ell \ge 2$ is fixed throughout.
The graph $\DblBase$ can now be formed from an imbalanced product.
Construct $\DblBase[p_1, p_2]$ from copies of $\Base[p_1]$ and $\Base[p_2]$ in exactly
the same way that $\DblBase[p]$ is constructed from two copies of $\Base[\sqrt{p}]$.  
The number of vertex layers in $\DblBase[p_1,p_2]$ is still $2\ell+1$ and each layer contains $p_1p_2$ vertices.
However, a node in an internal layer has $|\Labels[p_1]|$ neighbors in the previous layer and $|\Labels[p_2]|$
neighbors in the next layer (or vice versa), so the density of $\DblBase[p_1,p_2]$ is determined by $\max\{p_1,p_2\}$.

We define $\Graph_1^{\gamma}[p_1, p_2]$ to be a graph drawn from Lemma \ref{lem: hard girth graphs}, 
with $n=n(p_1,p_2)$ vertices and input/output ports $S,T$ selected with the following cardinality.
When $\gamma \ge 3$ is odd,
$$|S| = |T| = p_1 = p_2 = n^{(\gamma + 1)/(2\gamma)},$$
which, in Lemma \ref{lem: hard girth graphs}, corresponds to choosing $i = \frac{\gamma + 1}{2}$.  When $\gamma \ge 2$ is even,
$$|S| = p_1 = n^{(\gamma + 2)/(2\gamma)} \;\mbox{ and }\; |T| = p_2 = n^{1/2},$$
which corresponds to picking $i = \frac{\gamma + 2}{2}$ in Lemma \ref{lem: hard girth graphs}.
We define $\Pairs_1^{\gamma}$ to be the set of $\Omega(p_1p_2)$ 
node pairs in $\Graph_1^{\gamma}[p_1,p_2]$ from Lemma \ref{lem: girth pairs}.

We proceed as in Theorem \ref{thm: girth lower bounds}, but with a few critical differences.  Although $\Pairs_1^\gamma$-paths
through $\Graph_1^\gamma$ have length $\gamma$, there could be ``shortcuts'' between input ports and output ports not
covered by $\Pairs_1^\gamma$; the length of a shortcut might be as low as 1.
When forming $\Graph_k^\gamma$ we subdivide edges in $\DblBase$ as before, but have to make these
paths a factor $\gamma$ longer to sufficiently penalize paths that attempt to 
deviate far from the unique shortest path and thereby take advantage of many shortcuts elsewhere in $\Graph_k^\gamma$.
The imbalanced product $\DblBase[p_1,p_2]$ is only used in the formation of $\Graph_2[p]$, and only when $\gamma$ is even.

When $\gamma$ is odd, $\Graph_2^\gamma[p]$ is constructed from $\DblBase[p]=\DblBase[\sqrt{p},\sqrt{p}]$
and $\Graph_1^\gamma[\f{\sqrt{p}}{\Loss_\ell(\sqrt{p})},\f{\sqrt{p}}{\Loss_\ell(\sqrt{p})}]$ exactly as in Section~\ref{sect:lower-bound}, 
but replacing edges in $\DblBase$ by paths of length $\gamma(2\ell-1)$.
When $\gamma$ is even, to construct $\Graph_2^\gamma[p]$ we pick $p_1,p_2$ to have the ``right'' proportions
such that $p_1p_2 = p$.  The right proportions are dictated by the function $\Loss_\ell(\cdot)$ from the construction of $\Base$
and Lemma~\ref{lem: hard girth graphs}.
Let $p_1' = |\Labels[p_1]| \ge p_1/\Loss_\ell(p_1)$ and $p_2' = |\Labels[p_2]| \ge p_2/\Loss_\ell(p_2)$ be 
the number of edges connecting an internal node $u$ in $\DblBase[p_1,p_2]$ to previous/subsequent layers.
When forming $\Graph_2^\gamma$, each of these edges gets attached to a different input/output port of $\Graph_1^\gamma(u)$,
so we need $p_1' = (p_2')^{(\gamma+2)/\gamma}$.  In cases where $\Loss_\ell(p) = p^{o(1)}$, we can ignore the distinction between
$p_1$ and $p_1'$, and just set $p_1 = p^{\f{\gamma+2}{2\gamma+2}}$ and $p_2 = p^{\f{\gamma}{2\gamma+2}}$.
When $k\ge 3$, $\Graph_k^\gamma[p]$ is constructed from $\DblBase[p]$ and $\Graph_{k-1}^\gamma[\cdot]$ as before,
but subdivides edges into paths of length $\gamma(2\ell-1)^{k-1}$.

We now analyze the distances in $\Graph_k^\gamma$ of pairs in $\Pairs_k^\gamma$.

\begin{lemma} [Compare to Lemma \ref{lem: unique paths}] \label{lem: girth unique paths}
Fix a $(u_0, u_{2\ell}) \in \Pairs_k^{\gamma}$ whose unique shortest path in $\DblBase$ is $(u_0,u_1,\ldots,u_{2\ell})$.
The following hold.
\begin{itemize}
\item There is a unique $u_0$--$u_{2\ell}$ shortest path in $\Graph_k^\gamma$.  It has length 
$\gamma (2(k-1)\ell + 1)(2\ell - 1)^{k-1}$.

\item Any path from $u_0$ to $u_{2\ell}$ in $\Graph_k^{\gamma}$ that intersects some $\Graph_{k-1}^\gamma(u')$,
$u'\not\in \{u_1,\ldots,u_{2\ell-1}\}$, is at least $2(2\ell - 1)^{k-1}$ longer than the shortest path.
\end{itemize}
\end{lemma}

\begin{proof}
The proof is by induction.
When $k=1$, $(u_0,u_{2\ell})\in \Pairs_1^\gamma$ implies that
$\dist_{\Graph_1^{\gamma}}(u_0, u_{2\ell}) = \gamma$ (by Lemma \ref{lem: hard girth graphs}) 
and by the girth of $\Graph_1^\gamma$ the path is unique.

We now turn to the inductive step.
By the inductive hypothesis, there is only one shortest path that passes through 
$\Graph_{k-1}^\gamma(u_1),\ldots,\Graph_{k-1}^\gamma(u_{2\ell-1})$
and it has length
\[
d_k^\gamma = (2\ell-1)d_{k-1}^{\gamma} + (2\ell)\cdot \gamma(2\ell-1)^{k-1},
\]
which has a closed-form solution $d_k^\gamma = \gamma(2(k-1)\ell + 1)(2\ell - 1)^{k-1}$.

By Lemma~\ref{lem: unique paths}, any path from an input port of $\Graph_k^\gamma$ to an output port of $\Graph_k^\gamma$
passes through $(2\ell-1)^{k-1}$ copies of $\Graph_1^{\gamma}$.  Thus, the {\em minimum} length of such a path
is exactly $d_k^\gamma - (\gamma-1)(2\ell-1)^{k-1}$.  Consider a path that passes through some $\Graph_{k-1}^\gamma(u')$,
where $u'\not\in \{u_1,\ldots,u_{2\ell-1}\}$.  Since the shortest $u_0$--$u_{2\ell}$ path in $\DblBase$ is unique and $\DblBase$ is bipartite,
this path traverses at least two additional subdivided edges, each of length $\gamma (2\ell-1)^{k-1}$.  The length of such a path is therefore
at least 
\[
d_k^\gamma - (\gamma-1)(2\ell-1)^{k-1} + 2\cdot \gamma (2\ell-1)^{k-1},
\]
which is at least $d_k^\gamma + 2(2\ell-1)^{k-1}$.  Thus, the shortest $u_0$--$u_{2\ell}$ path in $\Graph_k^\gamma$ is unique.
\end{proof}

Since paths through $\Graph_1^\gamma$ overlap, we need to update the definition of a ``critical'' edge.

\begin{definition} [Compare to Definition \ref{def: critical}]
An edge $e$ is \emph{critical} for a pair $(u_0, u_{2\ell}) \in \Pairs_k^{\gamma}$ if it lies in a copy of $\Graph_1^{\gamma}$, the unique 
shortest $u_0$--$u_{2\ell}$ path in $\Graph_k^{\gamma}$ enters and leaves that copy of $\Graph_1^{\gamma}$ by some pair $(s, t) \in \Pairs_1^{\gamma}$, and we have $\phi(s, t) = e$. (Thus $(s, t)$ requires $e$ in $\Graph_1^{\gamma}$, and so $(u_0, u_{2\ell})$ requires $e$ in $\Graph_k^{\gamma}$.)
\end{definition}

\begin{lemma} \label{lem: girth path stretch}
Let $\tilde{\Graph}_k^{\gamma}$ be $\Graph_k^{\gamma}$ with all critical edges for $(u_0, u_{2\ell})$ removed.
Then $\dist_{\tilde{\Graph}_k^{\gamma}}(u_0, u_{2\ell}) \ge \dist_{\Graph_k^{\gamma}}(u_0, u_{2\ell}) + 2(2\ell - 1)^{k-1}$.
\end{lemma}
\begin{proof}
We proceed by induction.
In the base case of $k=1$, suppose a critical edge is removed for $(u_0, u_{2\ell})$.  Since $\dist_{\Graph_1^\gamma}(u_0,u_{2\ell}) = \gamma$
and $\Graph_1^\gamma$ has girth $2\gamma+2$, $\dist_{\tilde{\Graph}_1^\gamma}(u_0,u_{2\ell}) \ge \gamma+2$.
For the inductive step, let $Q$ be the shortest path from $u_0$ to $u_{2\ell}$ in $\Graph_k^{\gamma}$, 
and let $\tilde{Q}$ be the shortest path in the graph $\tilde{\Graph}_k^{\gamma}$.
Suppose first that $\tilde{Q}$ traverses the exact same copies of $\Graph_{k-1}^\gamma$ that $Q$ traverses.
In this case the claim follows from the inductive hypothesis: we accumulate $2(2\ell-1)^{k-2}$ additive stretch in
each of the $2\ell-1$ copies of $\Graph_{k-1}^\gamma$ traversed.
If $\tilde{Q}$ deviates and intersects some other copy of $\Graph_{k-1}^\gamma$, then by Lemma~\ref{lem: girth unique paths},
the additive stretch is at least $2(2\ell-1)^{k-1}$.
\end{proof}

\begin{lemma} \label{lem: girth distance equal}
Let $\tilde{\Graph}_k^{\gamma}$ be $\Graph_k^{\gamma}$ with the critical edges for all pairs in 
$\Pairs_k^{\gamma}$ except $(u_0, u_{2\ell})$ removed.  
Then $\dist_{\tilde{\Graph}_k^{\gamma}}(u_0, u_{2\ell}) = \dist_{\Graph_k^{\gamma}} (u_0, u_{2\ell})$.
\end{lemma}
\begin{proof}
Let $s, t$ be the input/output ports used by shortest paths from $u_0$ to $u_{2\ell}$ in any internal copy of $\Graph_1^{\gamma}$.
By Lemma \ref{lem: girth pairs}, the unique shortest path from $s$ to $t$ in $\Graph_1^{\gamma}$ does not include any edge 
$e$ for which $\phi(s',t') = e,  (s',t') \neq (s, t)$.
It follows that the distance from $s$ to $t$ is the same in $\Graph_k^{\gamma}$ and $\tilde{\Graph}_k^{\gamma}$.
\end{proof}

The final piece of the proof is exactly identical to the lower bound argument in Section \ref{sect:lower-bound}.
In particular, we define a family of $2^{|\Pairs_k^{\gamma}|}$ graphs by keeping/removing the critical edges for each pair in 
$\Pairs_k^{\gamma}$ in all possible combinations.
By Lemmas \ref{lem: girth path stretch} and \ref{lem: girth distance equal}, any two of these graphs will disagree on a 
pairwise distance $(u_0, u_{2\ell})$ by an additive $2(2\ell - 1)^{k-1}$.
By Lemma \ref{lem: girth unique paths}, we have 
$\dist_{\Graph_k^{\gamma}}(u_0, u_{2\ell}) = d = \gamma(2(k-1)\ell + 1)(2\ell - 1)^{k-1}$.
Thus, the additive stretch $2(2\ell-1)^{k-1}$ is roughly $\fr{2}{(\gamma(k-1))^{1-1/k}} \cdot d^{1-1/k}$.
If the stretch function of the distance oracle is $f(d) \le d + c d^{1-1/k}$ for sufficiently large $d$ 
and a sufficiently small constant $c < \fr{2}{(\gamma(k-1))^{1-1/k}}$, then it cannot map any of these graphs to the same bit-string.
This gives the stretch part of the lower bound claimed in Theorem \ref{thm: girth lower bounds}.
It remains only to compute the size of this graph family.
We have $2^{|\Pairs_k^{\gamma}|}$ distinct graphs, so we need to obtain a lower bound on $|\Pairs_k^\gamma|$.
In order to avoid tedious calculations let 
us assume that $\ell = p^{o(1)}$, so $\Loss_\ell(p) = p^{o(1)}$ as well.
In particular, $\nBase[p]$, $\nDblBase[p] = p^{1+o(1)}$, 
$\mDblBase[p_1,p_2] = p_1p_2(p_1+p_2)^{1-o(1)}$, 
and $|\Pairs(\DblBase[p_1,p_2])| = (p_1p_2)^{2-o(1)}$.
Letting $n_1^\gamma[p_1,p_2]$ be the number of vertices in $\Graph_1^\gamma[p_1,p_2]$
and $n_k^\gamma[p]$ be the number of vertices in $\Graph_k^\gamma[p]$, we have
\begin{align*}
n_1^\gamma[p_1, p_2]		&= (p_1 p_2)^{\gamma / (\gamma + 1)}\\
n_2^\gamma[p]				&= p^{1+o(1)} \cdot n_{1}^\gamma[p_1^{1-o(1)}, p_2^{1-o(1)}] + (\mDblBase[p_1,p_2])^{1+o(1)}\\
\intertext{Where $p_1,p_2 = \sqrt{p}$ if $\gamma$ is odd and $p_1 = p^{\f{\gamma+2}{2\gamma+2}}, p_2 = p^{\f{\gamma}{2\gamma+2}}$ if $\gamma$ is even. When $k\ge 3$,}
n_k^\gamma[p]				&= p^{1+o(1)} \cdot n_{k-1}^\gamma[p^{1/2 - o(1)}] + (\mDblBase[p])^{1+o(1)}.
\end{align*}
Whether $\gamma$ is even or odd, $n_{1}^\gamma[p_1^{1-o(1)}, p_2^{1-o(1)}] = (p_1p_2)^{\gamma/(\gamma+1) - o(1)}$.
The density of $\DblBase[p_1,p_2]$ is maximized when $p_1$ and $p_2$ are most imbalanced.  This occurs
when $\gamma=2$, $p_1 = p^{2/3}$ and $p_2 = p^{1/3}$, making $\mDblBase[p_1,p_2] = p^{5/3 - o(1)}$.
Thus, for any $\gamma\ge 2$, $n_2^\gamma[p] = p^{2 - \f{1}{\gamma+1} + o(1)}$.
By induction on $k$, $n_k^\gamma[p] = p^{2 - \f{1}{(\gamma+1)2^{k-2}} + o(1)}$.

By Lemma~\ref{lem: girth pairs}, $|\Pairs_1^\gamma[p_1,p_2]| = \Omega(p_1p_2)$.
The same inductive proof from Section~\ref{sect:lower-bound} shows that for 
any $k\ge 2$, $|\Pairs_k^\gamma[p]| = p^{2-o(1)}$.  
Expressed in terms of $n = n_k^\gamma[p]$, $p^{2-o(1)}$
is $n^{1+\f{1}{(\gamma+1)2^{k-1}-1} - o(1)}$.   
Theorem \ref{thm: girth lower bounds} follows.

\subsection{Matching Upper Bounds}\label{sect:girth-upper-bounds}

The subgraph $E_0'$ from Section~\ref{sect:new-upper-bounds} can be viewed as a {\em radius-1 clustering}
of the graph, obtained from the following procedure.
First, cluster centers $V_1 \subset V$ are sampled with probability $q_1$.  Each vertex
incident to $V_1$ is {\em clustered} and joins the cluster of one such adjacent $V_1$ vertex.
$E_0'$ contains a star spanning each cluster and all edges incident to unclustered vertices,
which number $O(n/q_1)$ in expectation.
Baswana et al.~\cite{BKMP10} observed that in graphs with
girth at least $2\gamma+1$, this procedure can be generalized to compute a radius-$\gamma$ clustering 
with similar properties.

\begin{theorem} [\cite{BKMP10}] \label{thm:bkmp-clustering}
Let $G=(V,E)$ be a graph with girth at least $2\gamma+1$.
Fix $q_1 < 1$ and let $V_1\subset V$ be obtained by sampling each element of $V$ with probability $q_1$.
Any $v\in V$ with $\dist(v,V_1)\le \gamma$ is {\em clustered} and joins the cluster of the closest $V_1$ vertex, breaking ties consistently.
Let $E_0^\gamma$ contain a radius-$\gamma$ tree on each cluster and all edges incident to unclustered vertices.
In expectation $|E_0^\gamma| = O(n/q_1^{1/\gamma})$.
\end{theorem}

We can use the $E_0^\gamma$ from Theorem~\ref{thm:bkmp-clustering} in lieu of $E_0$ in the construction
of Thorup-Zwick emulators~\cite{TZ06}.  The total size of the emulator 
$E^\gamma = E_0^\gamma \cup E_1 \cup\cdots \cup E_k$ is then
on the order of
\[
\f{n}{q_1^{1/\gamma}} + \f{nq_1^2}{q_2} + \cdots + \f{nq_{k-1}^2}{q_k} + (nq_k)^2.
\]
If we write $q_i$ as $n^{-g(i)}$, $g(i)$ must satisfy the following.
\begin{align*}
g(i) &= 2g(i-1) + \fr{g(1)}{\gamma}		& \mbox{(balancing $E_0^\gamma$ and $E_{i-1}$, for $i\in[2,k]$)}\\
	&= (\gamma+1)(2^{i-1}-1)\fr{g(1)}{\gamma}	& \mbox{(by induction)}
\end{align*}
Balancing the size of $E_k$ and $E_0^\gamma$ we have
\[
|E_k| = (nq_k)^2 = n^{2 - ((\gamma+1)2^k -2)\f{g(1)}{\gamma}}  = n^{1+\f{g(1)}{\gamma}} = |E_0^\gamma|,
\]
which is satisfied when $\f{g(1)}{\gamma} = \f{1}{(\gamma+1)2^k-1}$, implying the size of the emulator is
$O(kn^{1+\f{1}{(\gamma+1)2^k -1}})$.
The analysis of the emulator proceeds in exactly as in Section~\ref{sect:new-upper-bounds}, by bounding the quantities $\Succ{\ell}{i}$
and $\Fail{\ell}{i}$ inductively.  Substituting $E_0^\gamma$ for $E_0$ only affects the following base cases.
\begin{align*}
\Succ{\ell}{0} &= 0		& \mbox{ for all $\ell$}\\
\Fail{\ell}{0} &= \gamma	& \mbox{ for all $\ell$}
\end{align*}
This is justified by Theorem~\ref{thm:bkmp-clustering}.  Any path with length $\ell^0 = 1$ is a single edge, say $(u,v)$.
If $u$ is unclustered in $E_0^\gamma$ then $(u,v)\in E^\gamma$, 
$\dist_{E^\gamma}(u,v)=1$, and $(u,v)$ is \emph{complete}.
On the other hand, if $u$ is clustered then 
$\dist_{E^\gamma}(u,p_1(u)) \le \gamma$ and $(u,v)$ is \emph{incomplete}.
With these base cases it is straightforward to show the stretch function for $E^\gamma$ is
$f(d) = d + O(\gamma k d^{1-1/k})$.   For example, when $\gamma=k=2$ we see that every girth-5 graph
has an $(d+O(\sqrt{d}))$-emulator with size $O(n^{12/11})$. 

This emulator can be converted to a $(1+\epsilon,O(\gamma k/\epsilon)^{k-1})$-spanner 
by applying the same transformations from Section~\ref{sect:new-upper-bounds}, using 
Theorem~\ref{thm:girth-path-buying} in lieu of Theorem~\ref{thm:path-buying}.

\begin{theorem}\label{thm:girth-path-buying}
Let $G=(V,E),q_1,V_1,$ and $E_0^\gamma$ be as in Theorem~\ref{thm:bkmp-clustering}.
Suppose $V_2$ is obtained by sampling each element of $V$ with probability $q_2$, where $q_2 < q_1$.
There is an edge-set $\tilde{E}_1^\gamma$ with expected size $O(\gamma^2 n q_1^2/q_2)$
such that for $u,v\in V_1$ and $v\in \Ball_2(u)$,
\[
\dist_{E_0^\gamma \cup \tilde{E}_1^\gamma}(u,v) \le 2\gamma.
\]
\end{theorem}

\begin{remark}
The $\gamma^2$ factor arises from two parts of the path-buying algorithm's analysis that
depend on the cluster radii.  The {\em cost} of a path (number of missing edges) is at most the
number of clusters touching the path divided by the cluster diameter, $2\gamma$.  Once a
cluster-pair is charged we have their correct distance to within $+O(\gamma)$.  The path-buying
algorithm only charges this cluster pair again when the distance improves, so at most $O(\gamma)$ times.
\end{remark}

The base case values for $\Succ{\ell}{i}$ and $\Fail{\ell}{i}$ are updated as follows.  For all $\ell$,
\begin{align*}
\Succ{\ell}{0} &= 0			\\
\Fail{\ell}{0} &= \gamma		\\
\Succ{\ell}{1} &= 6\gamma		\\
\Fail{\ell}{1} &= \ell + 3\gamma	
\end{align*}
It is easy to check that these base cases increase $\Fail{\ell}{i} - \ell^i$ by a factor of $\gamma$.
It is for this reason that we use a slightly larger threshold $(r+2\gamma)^i$ when forming $E_i'$ in the following construction.
The spanner $S(k,r,\gamma)$ has the edge-set $E_0^\gamma \cup \tilde{E}_1^\gamma \cup E_2' \cup \cdots \cup E_k'$,
where $E_i'$ is obtained by replacing each (weighted) pair $(u,v)\in E_i$ with a shortest path $P(u,v)$, assuming $\dist(u,v)$ is sufficiently short.
\[
E_i' = \bigcup_{\substack{(u,v) \in E_i \,:\\\dist(u,v) \le (r + 2\gamma)^i}} P(u,v).
\]
It follows that the size of the spanner is on the order of\footnote{For simplicity we treat the $\gamma^2$ factor in $|\tilde{E}_1^\gamma|$ as a constant.}
\[
\f{n}{q_1^{1/\gamma}}  + \f{nq_1^2}{q_2} + \f{nq_2^2 r^2}{q_3} \cdots \f{nq_{k-1}^2r^{k-1}}{q_k} + (nq_k)^2r^k.
\]
Writing $q_i = n^{-g(i)}r^{-h(i)}$, $g(i)$ satisfies the same recurrence as before and $h(i)$ satisfies the following.
\begin{align*}
h(2)	&= 2h(1) + \fr{h(1)}{\gamma}		& \mbox{(balancing $E_0^\gamma$ and $\tilde{E}_1^\gamma$)}\\
h(i) 	&= 2h(i-1) + \fr{h(1)}{\gamma} - (i-1)				& \mbox{(balancing $E_0^\gamma$ and $E_{i-1}'$, $i\in [3,k]$)}.\\
	&= ((\gamma+1)2^{i-1} - 1)\fr{h(1)}{\gamma} - 3\cdot 2^{i-2} + (i+1)	& \mbox{(by induction, for $i \in [3,k]$)}
\end{align*}
Finally, we balance $E_0^\gamma$ and $E_k'$,
\[
|E_k'| = (nq_k)^2r^k = n^{2 - ((\gamma+1)2^k - 2)\f{g(1)}{\gamma}} r^{k - [((\gamma+1)2^{k}-2)\f{h(1)}{\gamma} - 3\cdot 2^{k-1} + 2(k+1)]} = n^{1+\f{g(1)}{\gamma}}r^{1+\f{h(1)}{\gamma}} = |E_0^\gamma|,
\]
by setting $g(1),h(1)$ as follows.
\begin{align*}
\f{g(1)}{\gamma} &= \f{1}{(\gamma+1)2^k -1}\\
\f{h(1)}{\gamma} &= \f{3\cdot 2^{k-1} - (k+2)}{(\gamma+1)2^k -1}.
\end{align*}
Thus, the size of the resulting spanner is $O(r^h kn^{1+\f{1}{(\gamma+1)2^k -1}})$, 
where $h=\f{3\cdot 2^{k-1} - (k+2)}{(\gamma+1)2^k -1}$.
For example, setting $k=\gamma=2$ and $r=\sqrt{D}$, this shows that every graph with girth 5 contains a subgraph that
functions like a $(d+O(\sqrt{d}))$-spanner for all $d\le D$, with size $O(r^{\f{2}{11}}n^{\f{12}{11}}) = O(D^{\f{1}{11}}n^{\f{12}{11}})$.

Theorem summarizes our emulator and spanners constructions for high-girth graphs.

\begin{theorem}\label{thm:girth-emul-span}
Let $G$ be a graph with girth at least $2\gamma+1$.  There is an additive $4\gamma$-emulator and
additive $6\gamma$-spanner for $G$ with size $O(n^{1+\f{1}{2\gamma+1}})$.
For any integer $k\ge 2$, there is an $(d+O(\gamma kd^{1-1/k}))$-emulator
for $G$ with size $O(kn^{1+\f{1}{(\gamma+1)2^k -1}})$ and a $(1+\epsilon, O(\gamma k/\epsilon)^{k-1})$-spanner
for $G$ with size $O((\gamma k/\epsilon)^h kn^{1+\f{1}{(\gamma+1)2^k -1}})$, where $h=\f{3\cdot 2^{k-1} - (k+2)}{(\gamma+1)2^k-1} < \f{3}{2(\gamma+1)}$.
\end{theorem}

\section{Lower Bounds on Shortcutting Digraphs}\label{sect:shortcut}

In this section we consider {\em directed} unweighted graphs $G=(V,E)$.  
Let $u\leadsto v$ be the reachability (transitive closure) relation for $G$, 
indicating a directed path from $u$ to $v$.
In 1992 Thorup~\cite{Thorup92} conjectured that for any directed graph $G=(V,E)$
there exists another $G' = (V,E')$ such that 
(i) $G$ and $G'$ have the same reachability relation ($\leadsto$),
(ii) $|E'|\le 2|E|$, and 
(iii) every $u\leadsto v$ is witnessed in $G'$ by a directed path with length $\poly(\log n)$;
this is called the {\em diameter} of $G'$.
Thorup's conjecture was confirmed for trees~\cite{Thorup92,Thorup97-par-shortcut,Chaz87}
and planar graphs~\cite{Thorup95}, but disproved in a strong form by Hesse~\cite{Hesse03},
who showed that there exists a $G$ with $n^{1+\epsilon}$ edges and diameter $n^{\delta(\epsilon)}$,
such that any shortcutting $G'$ with diameter $o(n^\delta)$ requires $\Omega(n^{2-\epsilon})$ edges.
In this section we give a simpler proof of Hesse's result---a refutation of Thorup's conjecture---by generalizing the construction of $\DblBase$ from Section~\ref{sect:lower-bound}.  

\subsection{The Construction}

Recall that $\Base[p]$ is parameterized by an integer $\ell \ge 2$.  Its vertex set is partitioned 
into $\ell+1$ layers of $p$ vertices; each vertex has $p/\Loss_\ell(p)$ edges leading to the next layer,
each of which is assigned a distinct label from the set $\Labels[p]$.  
Here $\Loss_\ell(p) = 2^{\Theta(\sqrt{\log\ell \cdot \log p})}$.
The set $\Pairs = \Pairs(\Base[p])$ consists of $p^2 / \Loss_\ell(p)$ pairs, each having a unique length-$\ell$ shortest path.  Each element $(u,v)\in \Pairs$ is generated by picking a vertex $u$ in the first layer and a label $a\in\Labels[p]$: $v$ is the vertex in the last layer reached by starting at $u$ and repeatedly following edges labeled $a$.
In this section we regard $\Base[p]$ as being a directed acyclic graph, with all edges oriented toward the higher numbered layer.

Rather than form $\DblBase[p]$ by taking the product of two copies of $\Base[\sqrt{p}]$, 
we take the product of $k$ copies of $\Base[p^{1/k}]$.   
Let the layers of $\Base[p^{1/k}]$ be $L_0,\ldots,L_\ell$.
The vertex set of
$\DblBase[p]$ is partitioned into layers $\Layer{0},\ldots,\Layer{k\ell}$.
If $q$ is written $ik+j$, where $i\le \ell, j<k$,  $\Layer{q}$ is the set $L_{i+1}^j \times L_i^{k-j}$.
A directed edge $(\nu,\nu')\in \Layer{q} \times \Layer{q+1}$ exists if 
$\nu$ and $\nu'$ only differ in their 
$j$th component and $(\nu[j],\nu'[j])$ is in the edge-set of $\Base[p^{1/k}]$.
In other words, a path from layer $\Layer{0}$ to layer $\Layer{q}$ in $\DblBase$ simulates
$k$ independent paths, from layer $L_0$ to $L_{i}$ in $k-j$ copies of $\Base$, and from
layer $L_0$ to $L_{i+1}$ in the remaining $j$ copies.
The pair-set $\Pairs(\DblBase[p])$ is defined as one might expect:
\[
\Pairs(\DblBase[p]) = \{(\nu,\nu') \;|\; (\nu[j],\nu'[j]) \in \Pairs(\Base), \mbox{ for each $0 \le j < k$}\}
\]
Thus, for any $(\nu,\nu')\in \Pairs(\DblBase[p])$, 
$\nu\leadsto\nu'$ is witnessed by a unique path having length $k\ell$
and the labels along this path form a periodic sequence
$(a_0,a_1,\ldots,a_{k-1},a_0,a_1,\ldots)$ for some $(a_0,\ldots,a_{k-1})\in (\Labels[p^{1/k}])^k$.
The size of the pair-set $\Pairs = \Pairs(\DblBase[p])$ is 
$|\Pairs| \ge (p^{2/k} / \Loss)^k = p^2 / \Loss^k$, where $\Loss = \Loss_\ell[p^{1/k}]$.

\begin{lemma}\label{lem:shortcut-diam}
The diameter of $G = \DblBase[p]$ is $k\ell$.  
Any graph $G'$ with the same transitive closure as $G$ and diameter
$k\ell/(k-1) - 1$ must have at least $|\Pairs|$ edges.
\end{lemma}

\begin{proof}
Let $(\nu,\nu')\in\Pairs$, $P$ be the unique $\nu$--$\nu'$ path in $G$ and
$P'$ be a $\nu$--$\nu'$ path in $G'$ having length strictly shorter than $k\ell/(k-1)$.
It must be that $P'$ contains an edge that shortcuts at least $k$ consecutive edges in $P$.
However, because the edge-labels along $P$ are periodic with length $k$, 
any length-$k$ subpath of $P$ uniquely identifies $P$.  Thus, no edge of $G'$ that shortcuts
$k$ or more edges in $G$ can be used by two distinct pairs in $\Pairs$.  It follows
that $G'$ contains at least $|\Pairs|$ edges or the diameter of $G'$ is at least $k\ell/(k-1)$.
\end{proof}

The number of vertices and edges in $\DblBase[p]$ is $n=(k\ell+1)p$ and $m=k\ell p^{1+1/k}/\Loss_\ell(p^{1/k})$, respectively.  Setting $\ell = p^{\delta/k}$ for some small $\delta>0$ we have 
$\Loss_\ell(p^{1/k}) = 2^{\Theta(\sqrt{\log\ell \cdot \log p^{1/k}})} = p^{\Theta(\sqrt{\delta}/k)}$,
so the density of the graph is roughly 
$p^{(1-\Theta(\sqrt{\delta}))/k} \approx n^{\frac{1-\Theta(\sqrt{\delta})}{k(1+\delta/k)}}$.  
By Lemma~\ref{lem:shortcut-diam}, 
in order to reduce the diameter to $p^{\delta/k} = O(n^{\frac{\delta}{k(1+\delta/k)}})$ we need to 
add $|\Pairs| = p^2 / \Loss^{k} = p^{2 - \Theta(\sqrt{\delta})}$ shortcuts.
By setting $\delta = O(\epsilon^2)$ to be sufficiently small and $k=\Omega(1/\epsilon)$,
we arrive at the same conclusion of Hesse~\cite{Hesse03}.

\begin{theorem} (\cite{Hesse03})
For any $\epsilon > 0$ there exists a $\delta = \delta(\epsilon)$, 
a directed graph $G$ with $n$ vertices, at most $n^{1+\epsilon}$ edges, and diameter $n^{\delta}$
with the following property.  Any graph $G'$ with the same transitive closure as $G$
and diameter $o(n^{\delta})$ must contain at least $n^{2-\epsilon}$ edges.
\end{theorem}

\section{Conclusion}\label{sect:conclusion}

In this paper, we characterized the optimal asymptotic behavior of sublinear additive stretch functions $f$ for spanners, emulators, or any graph compression scheme.
Roughly speaking, any representation using $n^{1+\f{1}{2^{k}-1}-\delta}$ bits (for any $\delta > 0$) must have stretch function 
$f(d) = d + \Omega(d^{1-\f{1}{k}})$.
Previous constructions of sublinear additive emulators~\cite{TZ06} and $(1+\epsilon,\beta)$-spanners (\cite{EP04,TZ06} 
and the construction of Section~\ref{sect:new-upper-bounds})
show that neither the exponent $1+\f{1}{2^{k}-1}$ nor additive stretch $\Omega(d^{1-\f{1}{k}})$ can be improved, for any $k$, 

The main distinction between $(1+\epsilon,\beta)$-spanners~\cite{EP04,TZ06,Pettie-Span09} and
sublinear additive emulators/spanners~\cite{TZ06,Pettie-Span09,Chechik13} is that constructions
of the former take $\epsilon$ as a parameter (which affects the size of the spanner) whereas
the latter have $(1+\epsilon,\beta)$-stretch for all $\epsilon$, that is, $\epsilon$ can be chosen in the analysis.
An interesting open question is whether one can match the size-stretch tradeoff of 
Thorup and Zwick's optimal {\em emulators}~\cite{TZ06} with a {\em spanner}.
(Constructions in \cite{Pettie-Span09, Chechik13} are off from \cite{TZ06} (and our lower bounds) by a polynomial factor.)  It would be possible to construct such spanners given a \emph{pairwise} spanner
with a sublinear additive stretch function.  For example, when $S\subset V$ with $|S|=\Omega(n^{4/7})$
and $P=S\times S$, does there exist a pairwise spanner for $P$ with stretch $d+O(\sqrt{d})$ and size 
$O(|P|)$?  If such an object existed, we would immediately have an optimal $(d+O(\sqrt{d}))$-spanner with size $O(n^{8/7})$; see~\cite{TZ06,Pettie-Span09}.

Our lower bounds match the existing upper bounds in the distance regime 
$2^{\Omega(k)} \ll d < n^{o(1)}$, while they say nothing when $d = 2^{O(k)}$ and they are weaker when $d = n^{\Omega(1)}$.
An interesting open problem is to understand the sparseness-stretch tradeoffs available when $d=O(2^k)$ is tiny 
(see~\cite{EP01,BKMP10,Parter14}) and when $d=n^{\Omega(1)}$ is very large~\cite{BCE06,BodwinW16}.\\

\paragraph{Acknowledgments.}
We are grateful to Virginia Vassilevska Williams for useful technical discussions and for advice about the directions taken by this paper.
We thank Michael Elkin for proposing the question of finding a lower bound hierarchy for mixed spanners (as shown here), as well as observing the corresponding upper bounds.


\begin{thebibliography}{10}

\bibitem{AbboudB16}
A.~Abboud and G.~Bodwin.
\newblock The $4/3$ additive spanner exponent is tight.
\newblock In {\em Proceedings 48th Annual {ACM} Symposium on Theory of
  Computing ({STOC})}, pages 351--361, 2016.

\bibitem{AbboudB16-SODA}
A.~Abboud and G.~Bodwin.
\newblock Error amplification for pairwise spanner lower bounds.
\newblock In {\em Proceedings 27th Annual {ACM-SIAM} Symposium on Discrete
  Algorithms (SODA)}, pages 841--854, 2016.

\bibitem{AbrahamG11}
I.~Abraham and C.~Gavoille.
\newblock On approximate distance labels and routing schemes with affine
  stretch.
\newblock In {\em Proceedings 25th International Symposium on Distributed
  Computing ({DISC})}, pages 404--415, 2011.

\bibitem{Agarwal14}
R.~Agarwal.
\newblock The space-stretch-time tradeoff in distance oracles.
\newblock In {\em Proceedings 22nd Annual European Symposium Algorithms
  ({ESA})}, pages 49--60, 2014.

\bibitem{AgarwalG13}
R.~Agarwal and P.~B. Godfrey.
\newblock Distance oracles for stretch less than 2.
\newblock In {\em Proceedings 24th Annual {ACM-SIAM} Symposium on Discrete
  Algorithms ({SODA})}, pages 526--538, 2013.

\bibitem{ACIM99}
D.~Aingworth, C.~Chekuri, P.~Indyk, and R.~Motwani.
\newblock Fast estimation of diameter and shortest paths (without matrix
  multiplication).
\newblock {\em SIAM J.~Comput.}, 28(4):1167--1181, 1999.

\bibitem{Alon01}
N.~Alon.
\newblock Testing subgraphs in large graphs.
\newblock In {\em Proceedings 42nd IEEE Symposium on Foundations of Computer
  Science (FOCS)}, pages 434--441, 2001.

\bibitem{Althofer+93}
I.~Alth\"{o}fer, G.~Das, D.~Dobkin, D.~Joseph, and J.~Soares.
\newblock On sparse spanners of weighted graphs.
\newblock {\em Discrete and Computational Geometry}, 9:81--100, 1993.

\bibitem{BKMP10}
S.~Baswana, T.~Kavitha, K.~Mehlhorn, and S.~Pettie.
\newblock Additive spanners and $(\alpha,\beta)$-spanners.
\newblock {\em ACM Trans. on Algorithms}, 2009.

\bibitem{BaswanaS07}
S.~Baswana and S.~Sen.
\newblock A simple and linear time randomized algorithm for computing sparse
  spanners in weighted graphs.
\newblock {\em J.~Random Structures and Algs.}, 30(4):532--563, 2007.

\bibitem{Behrend46}
F.~Behrend.
\newblock On sets of integers which contain no three terms in arithmetic
  progression.
\newblock {\em Proc. Nat. Acad. Sci.}, 32:331--332, 1946.

\bibitem{Benson66}
C.~Benson.
\newblock Minimal regular graphs of girth eight and twelve.
\newblock {\em Canadian Journal of Mathematics}, 18:1091--1094, 1966.

\bibitem{BodwinW15}
G.~Bodwin and V.~Vassilevska Williams.
\newblock Very sparse additive spanners and emulators.
\newblock In {\em Proceedings 2015 Conference on Innovations in Theoretical
  Computer Science ({ITCS})}, pages 377--382, 2015.

\bibitem{BodwinW16}
G.~Bodwin and V.~Vassilevska Williams.
\newblock Better distance preservers and additive spanners.
\newblock In {\em Proceedings 27th Annual {ACM-SIAM} Symposium on Discrete
  Algorithms ({SODA})}, pages 855--872, 2016.

\bibitem{Bollobas78}
B.~Bollob{\'a}s.
\newblock {\em Extremal graph theory}, volume~11 of {\em London Mathematical
  Society Monographs}.
\newblock Academic Press Inc. [Harcourt Brace Jovanovich Publishers], London,
  1978.

\bibitem{BCE06}
B.~Bollob{\'a}s, D.~Coppersmith, and M.~Elkin.
\newblock Sparse subgraphs that preserve long distances and additive spanners.
\newblock {\em SIAM J. Discr. Math.}, 9(4):1029--1055, 2006.

\bibitem{BondyS74}
J.~Bondy and M.~Simonovits.
\newblock Cycles of even length in graphs.
\newblock {\em J. Combinatorial Theory, Series B}, 16:97--105, 1974.

\bibitem{Brown66}
W.~G. Brown.
\newblock On graphs that do not contain a {T}homsen graph.
\newblock {\em Canad. Math. Bull.}, 9:281--285, 1966.

\bibitem{Chaz87}
B.~Chazelle.
\newblock Computing on a free tree via complexity-preserving mappings.
\newblock {\em Algorithmica}, 2(3):337--361, 1987.

\bibitem{Chechik13}
S.~Chechik.
\newblock New additive spanners.
\newblock In {\em Proceedings 24th Annual {ACM}-{SIAM} Symposium on Discrete
  Algorithms (SODA)}, pages 498--512, 2013.

\bibitem{Chechik15}
S.~Chechik.
\newblock Approximate distance oracles with improved bounds.
\newblock In {\em Proceedings 47th Annual {ACM} on Symposium on Theory of
  Computing ({STOC})}, pages 1--10, 2015.

\bibitem{Cohen97}
E.~Cohen.
\newblock Using selective path-doubling for parallel shortest-path
  computations.
\newblock {\em Journal of Algorithms}, 22(1):30--56, 1997.

\bibitem{Cohen00}
E.~Cohen.
\newblock Polylog-time and near-linear work approximation scheme for undirected
  shortest-paths.
\newblock {\em J.~ACM}, 47:132--166, 2000.

\bibitem{CohenP10}
H.~Cohen and E.~Porat.
\newblock On the hardness of distance oracle for sparse graph.
\newblock {\em CoRR}, abs/1006.1117, 2010.

\bibitem{CE06}
D.~Coppersmith and M.~Elkin.
\newblock Sparse source-wise and pair-wise preservers.
\newblock {\em SIAM J.~Discrete Mathematics}, 20(2):463--501, 2006.

\bibitem{CyganGK13}
M.~Cygan, F.~Grandoni, and T.~Kavitha.
\newblock On pairwise spanners.
\newblock In {\em Proceedings 30th International Symposium on Theoretical
  Aspects of Computer Science ({STACS})}, pages 209--220, 2013.

\bibitem{DHZ00}
D.~Dor, S.~Halperin, and U.~Zwick.
\newblock All-pairs almost shortest paths.
\newblock {\em SIAM J.~Comput.}, 29(5):1740--1759, 2000.

\bibitem{EN16}
M.~Elkin and O.~Neiman.
\newblock Hopsets with constant hopbound, and applications to approximate
  shortest paths.
\newblock In {\em Proc. 57th IEEE Symposium on Foundations of Computer Science
  (FOCS), to appear}, 2016.

\bibitem{EP01}
M.~Elkin and D.~Peleg.
\newblock {$(1+\epsilon, \beta)$}-{spanner} constructions for general graphs.
\newblock In {\em Proc.~33rd Annual {ACM} Symposium on Theory of Computing
  ({STOC})}, pages 173--182, 2001.

\bibitem{EP04}
M.~Elkin and D.~Peleg.
\newblock $(1+\epsilon,\beta)$-spanner constructions for general graphs.
\newblock {\em SIAM J.~Comput.}, 33(3):608--631, 2004.

\bibitem{ElkinP16}
M.~Elkin and S.~Pettie.
\newblock A linear-size logarithmic stretch path-reporting distance oracle for
  general graphs.
\newblock {\em {ACM} Trans. Algorithms}, 12(4):50, 2016.

\bibitem{Erdos63}
P.~Erd\H{o}s.
\newblock Extremal problems in graph theory.
\newblock In {\em Theory of Graphs and its Applications (Proc. Sympos.
  Smolenice, 1963), pages 29--36. Publ. House Czechoslovak Acad. Sci., Prague},
  1963.

\bibitem{ErdosRS66}
P.~Erd{\H{o}}s, A.~R{\'e}nyi, and V.~T. S{\'o}s.
\newblock On a problem of graph theory.
\newblock {\em Studia Sci. Math. Hungar.}, 1:215--235, 1966.

\bibitem{Hesse03}
W.~Hesse.
\newblock Directed graphs requiring large numbers of shortcuts.
\newblock In {\em Proceedings 14th Annual {ACM-SIAM} Symposium on Discrete
  Algorithms (SODA)}, pages 665--669, 2003.

\bibitem{Kavitha15}
T.~Kavitha.
\newblock New pairwise spanners.
\newblock In {\em Proceedings 32nd International Symposium on Theoretical
  Aspects of Computer Science ({STACS})}, pages 513--526, 2015.

\bibitem{KavithaV15}
T.~Kavitha and N.~M. Varma.
\newblock Small stretch pairwise spanners and approximate d-preservers.
\newblock {\em SIAM J.~Discrete Mathematics}, 29(4):2239--2254, 2015.

\bibitem{KS97}
P.~N. Klein and S.~Subramanian.
\newblock A randomized parallel algorithm for single-source shortest paths.
\newblock {\em J.~Algor.}, 25(2):205--220, 1997.

\bibitem{Knudsen14}
M.~B.~T. Knudsen.
\newblock Additive spanners: {A} simple construction.
\newblock In {\em Proceedings 14th Scandinavian Symposium and Workshops on
  Algorithm Theory ({SWAT})}, pages 277--281, 2014.

\bibitem{LazebnikU93}
F.~Lazebnik and V.~A. Ustimenko.
\newblock New examples of graphs without small cycles and of large size.
\newblock {\em European J. of Combinatorics}, 14:445--460, 1993.

\bibitem{LazebnikUW95}
F.~Lazebnik, V.~A. Ustimenko, and A.~J. Woldar.
\newblock A new series of dense graphs of high girth.
\newblock {\em Bulletin of the AMS}, 32(1):73--79, 1995.

\bibitem{LazebnikUW96}
F.~Lazebnik, V.~A. Ustimenko, and A.~J. Woldar.
\newblock A characterization of the components of the graphs $d(k,q)$.
\newblock {\em Discrete Mathematics}, 157(1--3):271--283, 1996.

\bibitem{Matousek96}
J.~Matou{\v{s}}ek.
\newblock On the distortion required for embedding finite metric spaces into
  normed spaces.
\newblock {\em Israel J. Math.}, 93:333--344, 1996.

\bibitem{Parter14}
M.~Parter.
\newblock Bypassing {E}rd{\H{o}}s' girth conjecture: Hybrid stretch and
  sourcewise spanners.
\newblock In {\em Proceedings 41st International Colloquium Automata,
  Languages, and Programming ({ICALP})}, pages 608--619, 2014.

\bibitem{PatrascuR14}
M.~P{\v a}tra{\c s}cu and L.~Roditty.
\newblock Distance oracles beyond the {T}horup-{Z}wick bound.
\newblock {\em SIAM J.~Comput.}, 43(1):300--311, 2014.

\bibitem{PatrascuRT12}
M.~Patrascu, L.~Roditty, and M.~Thorup.
\newblock A new infinity of distance oracles for sparse graphs.
\newblock In {\em Proceedings 53rd Annual {IEEE} Symposium on Foundations of
  Computer Science ({FOCS})}, pages 738--747, 2012.

\bibitem{PS89}
D.~Peleg and A.~A. Schaffer.
\newblock Graph spanners.
\newblock {\em Journal of Graph Theory}, 13:99--116, 1989.

\bibitem{Pettie-Span09}
S.~Pettie.
\newblock Low distortion spanners.
\newblock {\em ACM Transactions on Algorithms}, 6(1), 2009.

\bibitem{PoratR13}
E.~Porat and L.~Roditty.
\newblock Preprocess, set, query!
\newblock {\em Algorithmica}, 67(4):516--528, 2013.

\bibitem{Reiman58}
I.~Reiman.
\newblock \"{U}ber ein {P}roblem von {K}. {Z}arankiewicz.
\newblock {\em Acta. Math. Acad. Sci. Hungary}, 9:269--273, 1958.

\bibitem{RTZ05}
L.~Roditty, M.~Thorup, and U.~Zwick.
\newblock Deterministic constructions of approximate distance oracles and
  spanners.
\newblock In {\em Proc. 32nd Int'l Colloq. on Automata, Lang., and
  Prog.~({ICALP})}, pages 261--272, 2005.

\bibitem{SS99}
H.~Shi and T.~H. Spencer.
\newblock Time-work tradeoffs of the single-source shortest paths problem.
\newblock {\em Journal of Algorithms}, 30(1):19--32, 1999.

\bibitem{SommerVY09}
C.~Sommer, E.~Verbin, and W.~Yu.
\newblock Distance oracles for sparse graphs.
\newblock In {\em Proceedings 50th Annual {IEEE} Symposium on Foundations of
  Computer Science ({FOCS})}, pages 703--712, 2009.

\bibitem{Thorup92}
M.~Thorup.
\newblock On shortcutting digraphs.
\newblock In {\em Proceedings 18th International Workshop on Graph Theoretic
  Concepts in Computer Science (WG)}, pages 205--211, 1992.

\bibitem{Thorup95}
M.~Thorup.
\newblock Shortcutting planar digraphs.
\newblock {\em Combinatorics, Probability {\&} Computing}, 4:287--315, 1995.

\bibitem{Thorup97-par-shortcut}
M.~Thorup.
\newblock Parallel shortcutting of rooted trees.
\newblock {\em J. Algorithms}, 23(1):139--159, 1997.

\bibitem{TZ05}
M.~Thorup and U.~Zwick.
\newblock Approximate distance oracles.
\newblock {\em J.~ACM}, 52(1):1--24, 2005.

\bibitem{TZ06}
M.~Thorup and U.~Zwick.
\newblock Spanners and emulators with sublinear distance errors.
\newblock In {\em Proc.~17th {ACM}-{SIAM} Symposium on Discrete Algorithms
  ({SODA})}, pages 802--809, 2006.

\bibitem{Tits59}
J.~Tits.
\newblock Sur la trialit\'{e} et certains groupes qui s'en d\'{e}duisent.
\newblock {\em Publ. Math. I.H.E.S.}, 2:14--20, 1959.

\bibitem{UY91}
J.~D. Ullman and M.~Yannakakis.
\newblock High-probability parallel transitive-closure algorithms.
\newblock {\em SIAM J. Comput.}, 20(1):100--125, 1991.

\bibitem{Wenger91}
R.~Wenger.
\newblock Extremal graphs with no {$C\sp 4$}s, {$C\sp 6$}s, or {$C\sp {10}$}s.
\newblock {\em J.~Combin. Theory Ser.~B}, 52(1):113--116, 1991.

\bibitem{WoldarU93}
A.~Woldar and V.~Ustimenko.
\newblock An application of group theory to extremal graph theory.
\newblock In {\em Group theory, Proceedings of the Ohio State-Denison
  Conference}, pages 293--298, River Edge, NJ, 1993. World Sci. Publishing.

\bibitem{Woodruff10}
D.~P. Woodruff.
\newblock Additive spanners in nearly quadratic time.
\newblock In {\em Proceedings 37th International Colloquium on Automata,
  Languages and Programming (ICALP)}, pages 463--474, 2010.

\end{thebibliography}


\end{document}